\documentclass[aip,jmp,amsmath,amssymb,preprint]{revtex4-2}

\input{glyphtounicode}
\usepackage{amsthm}
\usepackage{array}
\usepackage{booktabs}
\usepackage{enumitem}
\usepackage{placeins}
\usepackage{microtype}
\usepackage{setspace}
\usepackage{graphicx}
\usepackage{tikz}
\usetikzlibrary{arrows.meta}
\usepackage{hyperref}
\hypersetup{
  colorlinks=true,linkcolor=blue,citecolor=blue,urlcolor=blue,
  pdftitle={A Finite-Lattice Model from a Reciprocal Cost Action: Spectral and Reflection-Positivity Properties},
  pdfauthor={Jonathan Washburn and Megan Simons},
  pdfsubject={mathematical physics; lattice statistical mechanics; reflection positivity}
}

\newtheorem{theorem}{Theorem}[section]
\newtheorem{lemma}[theorem]{Lemma}
\newtheorem{proposition}[theorem]{Proposition}
\newtheorem{corollary}[theorem]{Corollary}
\newtheorem{definition}[theorem]{Definition}
\newtheorem{remark}[theorem]{Remark}
\newtheorem{openproblem}[theorem]{Open Problem}

\newcommand{\CBcond}[2]{\mathrm{CB}(#1,#2)}

\begin{document}

\title{A Finite-Lattice Model from a Reciprocal Cost Action:
Spectral and Reflection-Positivity Properties}

\author{Jonathan Washburn}
\email{jon@recognitionphysics.org}
\affiliation{Recognition Physics Institute, Austin, TX, USA}

\author{Megan Simons}
\email[Corresponding author: ]{msimons@recognitionphysics.org}
\affiliation{Recognition Physics Institute, Austin, TX, USA}

\date{June 2026}

\begin{abstract}
We study the finite-lattice statistical-mechanical model whose
nearest-neighbor bond potential is the reciprocal cost
$J(e^\varepsilon)=\cosh\varepsilon-1$, selected by the d'Alembert
functional equation under the stated regularity and calibration
assumptions.  The structural inputs are stated explicitly; once they are
fixed, the analysis is rigorous mathematics about the bond action
$V(\Delta\phi)=\cosh(\Delta\phi)-1$ on finite boxes in
$\mathbb Z^3\times\mathbb Z/8\mathbb Z$.  Our main result pairs a negative
and a positive statement about reflection positivity.  For the continuous
noncompact model the natural temporal kernel $K(u)=\exp[-(\cosh u-1)]$
fails the Bochner positive-definiteness test (an interval-certified
quadrature gives $\widetilde K(3)<0$), so the standard Bochner route to
Osterwalder--Schrader reflection positivity is obstructed.  For a
finite-alphabet variant, with field values restricted to a finite
symmetric set $\Phi=v_0\{-N,\ldots,N\}$, reflection positivity holds
whenever the finite crossing-bond Toeplitz matrix
$(K_{\Phi(v_0,N)})_{a,b}:=K(b-a)$, $a,b\in\Phi$, is positive semidefinite; for
$v_0\in\{1.2,1.5,2.5\}$ this is discharged by a rigorous
diagonal-dominance certificate uniform in $N$, and the associated
one-step transfer operator is then positive and self-adjoint in an
explicit reflection-positivity inner product.  These finite-volume
results do not provide a continuum Wightman theory,
Osterwalder--Schrader reconstruction, LSZ scattering, or a continuum
mass gap.  Ancillary algebraic observations appear in
Appendix~\ref{app:ancillary} and do not enter the reflection-positivity
results.
\end{abstract}


\maketitle

\section{Introduction}
\label{sec:intro}

The main mathematical object in this paper is a fixed nonlinear
nearest-neighbor lattice action:
\[
  V(\Delta\phi)=\cosh(\Delta\phi)-1 .
\]
Unlike a phenomenological lattice action chosen from a family of
potentials, this bond term is selected by the d'Alembert functional
equation for a reciprocal cost, under the stated regularity and
calibration assumptions.  We analyze the finite-lattice spectral and
reflection-positivity structure attached to this action.  Throughout,
the object of study is a classical Euclidean lattice
statistical-mechanical model, and quantum-mechanical vocabulary appears
only as interpretation: no Osterwalder--Schrader reconstruction to a
quantum theory is established (Section~\ref{sec:rp-transfer}).

The paper's central contribution is a paired negative and positive
result for this action.  On the negative side, the natural noncompact
temporal kernel $K(u)=\exp[-(\cosh u-1)]$ is not positive
definite: an interval-certified quadrature gives $\widetilde K(3)<0$
(Lemma~\ref{thm:fourier-negative}), so the standard Bochner route to
Osterwalder--Schrader reflection positivity is obstructed for the
continuous model.  On the positive side, restricting the field to a
finite symmetric alphabet yields reflection positivity for this
finite-alphabet variant whenever a finite crossing-bond Toeplitz matrix
is positive semidefinite (Theorem~\ref{thm:rp-discrete}); at the explicit spacings
$v_0\in\{1.2,1.5,2.5\}$ this hypothesis is discharged by a
diagonal-dominance certificate uniform in the alphabet size $N$
(Theorem~\ref{thm:gershgorin}, Corollary~\ref{cor:rp-verified}), and the
induced one-step transfer operator is positive self-adjoint in an
explicit reflection-positivity inner product
(Corollary~\ref{cor:transfer-psd}).  The remaining spectral and
algebraic facts collected here are elementary or standard and are
recorded for completeness.

Recognition Science
(RS)~\cite{Washburn2026Axioms,Washburn2026Coercive,PardoGuerra2026}
supplies the cost functional and the structural inputs used below:
observables are represented as positive comparison ratios
$x\in\mathbb R_{>0}$, the Recognition Composition Law fixes the
reciprocal cost, and the period~$8$, the value~$\varphi$, and the
spatial dimension $D=3$ are taken as model-defining inputs.  Once these
data are specified, the results proved here are finite-lattice or
finite-dimensional statements about the fixed bond action
$\cosh(\Delta\phi)-1$.

In previous work~\cite{Washburn2026Axioms,Washburn2026Coercive} we
introduced the \emph{Recognition Composition Law} (RCL),
\begin{equation}\label{eq:rcl}
  J(xy) + J(x/y) \;=\; 2\,J(x)\,J(y) + 2\,J(x) + 2\,J(y),
\end{equation}
as the composition axiom for joint cost.  Under the substitution
$H:=1+J\circ\exp$ the RCL is the d'Alembert functional equation, whose
calibrated continuous solutions are classical
(Acz\'el~\cite{Aczel1966}); it therefore uniquely determines the
canonical reciprocal cost
\begin{equation}\label{eq:jcost}
  J(x) = \tfrac{1}{2}(x + x^{-1}) - 1
\end{equation}
under the calibration A1--A3, as carried out in
Section~\ref{sec:forcing} (step~T5).  The bilinear combiner on the right
of~\eqref{eq:rcl} is itself the unique symmetric polynomial combiner of
degree at most two compatible with this
composition~\cite{Washburn2026DAlembert}, so the form of the law is fixed
by symmetry and degree rather than chosen ad hoc.
Separately from~\eqref{eq:jcost}, the RS framework also fixes the golden
ratio $\varphi=(1+\sqrt{5})/2$ used
below~\cite{Washburn2026Axioms,PardoGuerra2026}.

The paired Bochner failure and discrete-field reflection positivity of
Section~\ref{sec:rp-transfer} constitute the central result; the cyclic-shift
representation and action identity in Section~\ref{sec:complex} are
supporting facts, and the ancillary algebraic observations in
Appendix~\ref{app:ancillary} are not used in the reflection-positivity
argument.

\subsection{Related work}
\label{sec:related}

Most individual ingredients below are standard, and we attribute them
explicitly; the novelty is their assembly around the fixed
$\cosh-1$ bond action and, concretely, the Bochner failure of the
noncompact cosh kernel together with the discrete-field positive result
and its spacing threshold.  The functional equation~\eqref{eq:rcl} is the d'Alembert
equation, whose continuous solutions are classical
(Acz\'el~\cite{Aczel1966}); the canonical reciprocal cost~\eqref{eq:jcost}
is the corresponding classification, with a log-calibration
axiom~\cite{Washburn2026Axioms}; the admissible symmetric polynomial
combiners that reduce to d'Alembert form are classified
in~\cite{Washburn2026DAlembert}; and coercive projection methods for
finite ledger data appear in~\cite{Washburn2026Coercive}.  Diagonalizing
a cyclic shift by the discrete Fourier transform is textbook linear
algebra; the closed form of $\widetilde K$ as a modified Bessel function
of imaginary order follows from a standard integral
representation~\cite{Watson1944,DLMF}.  The reflection-positivity
machinery (Osterwalder--Schrader
axioms~\cite{OS1973,OS1975}, the product/Gram argument for reflection
positivity of nearest-neighbor lattice
models~\cite{FIS1978,Glimm1987,Biskup2009}, lattice transfer
matrices~\cite{OsterwalderSeiler1978}, and the general theory of lattice
systems, Gibbs measures, and noncompact
spins~\cite{Simon1993,Georgii2011,FriedliVelenik2017,LebowitzPresutti1976})
is
likewise standard.
The temporal kernel $K(u)=\exp[-(\cosh u-1)]$ is noncompact in field
space; analytical issues for noncompact lattice spin systems are
classical~\cite{LebowitzPresutti1976,Simon1993,Georgii2011}, and the discrete-field
restriction in Section~\ref{subsec:discrete-rescue} is the device by
which the present construction sidesteps the failure of the standard
Bochner positivity criterion for the crossing-bond kernel.  Our contribution on
this point is to record that the route fails for the cosh kernel
($\widetilde K(\xi)<0$ near $\xi\approx 3$) and to give a rigorous
diagonal-dominance certificate for the finite-alphabet variant at
explicit spacings.
For the sector measure, Appendix~\ref{sec:born} proves an
extension/consistency result: under the stated two-branch calibration,
$\sum_k|\psi_k|^2$ is the unique additive, phase-invariant,
coordinate-symmetric measure.  It does not derive the Born rule from
$J$-cost Gibbs weights.  Derivations of the Born rule
from weaker axioms are a distinct and substantial literature: Gleason's
theorem~\cite{Gleason1957}, the measurement-theoretic
approach~\cite{Busch2003}, Zurek's envariance
argument~\cite{Zurek2005}, and the decision-theoretic
Deutsch--Wallace program~\cite{Wallace2012}.  Our assumed two-branch
calibration is logically stronger than the hypotheses of those works and
should not be read as competing with them.
Proofs are given in the text.  Selected elementary facts are also
formalized in Lean~4 (Appendix~\ref{app:lean}).

\subsection{Notation and conventions}
\label{sec:notation}

\begin{itemize}[nosep]
  \item $\mathbb R_{>0}=(0,\infty)$.
  \item $\varphi=(1+\sqrt5)/2$.
  \item Unit system.  We work in units
        $\text{l}_0=\tau_0=1$, so $c=\text{l}_0/\tau_0=1$ and the action
        $S$ is dimensionless.  The action-to-phase constant
        $\hbar_{\rm RS}$ in $e^{iS/\hbar_{\rm RS}}$ enters no theorem below
        and may be set to $1$.
  \item The 8-tick index set is $\{0,1,\ldots,7\}$ with arithmetic mod~8.
  \item A hat, as in $\widehat f$, denotes the finite DFT-8 transform of
        an 8-mode signal.  A tilde, as in $\widetilde K$, denotes the
        continuous Fourier transform of a kernel on $\mathbb R$.
  \item For a lattice field $\phi$, $\Delta_\mu\phi(x):=\phi(x+\hat\mu)-\phi(x)$.
  \item Normalized 8-mode signals satisfy $\sum_{k=0}^7|\psi_k|^2=1$.
  \item ``Structural input'' denotes a hypothesis established in the RS
        framework and taken as given here, rather than derived
        from~\eqref{eq:jcost} alone (T6--T8).
  \item ``External theorem'' invokes a standard analytic result from the
        literature, with hypotheses verified in Lean or in the text.
\end{itemize}

\subsection{Roadmap}

Section~\ref{sec:related} situates the construction in the lattice
statistical-mechanics literature, and Section~\ref{sec:forcing} records
the model's provenance, separating the analytically fixed cost law~(T5),
with its elementary corollaries, from the structural inputs (T6--T8).
Section~\ref{sec:complex} records the cyclic-shift spectral
representation and the action identity
$J(e^\varepsilon)=\cosh\varepsilon-1$, and
Section~\ref{sec:rp-transfer} (the bulk of the paper) develops
reflection positivity: the noncompact Bochner failure, the
conditional discrete-field criterion and its diagonal-dominance
certificate, and the finite-box transfer kernel.
Section~\ref{sec:limits} collects what is established and what remains
open, including Open Problem~\ref{op:continuum-scaling} on continuum
scaling, and Section~\ref{sec:conclusion} concludes.
Appendix~\ref{app:fourier-cert} gives the interval certificate for the
sign of $\widetilde K(3)$; Appendices~\ref{app:reproducibility}
and~\ref{app:lean} collect reproducibility and formal-verification
information; and Appendix~\ref{app:ancillary} records ancillary algebraic
observations not used in the reflection-positivity results.

\section{Structural Inputs: Cost Law and Discrete Substrate}
\label{sec:forcing}

The object of study is a classical Euclidean lattice model on
$\mathbb Z^3\times\mathbb Z/8\mathbb Z$ with nearest-neighbor bond action
$\cosh(\Delta\phi)-1$.  This section records its provenance: the bond
action is fixed analytically by the cost functional $J$ (the
d'Alembert classification, T5), while the discrete geometry (spatial dimension $D=3$ and
temporal period $8$) together with the scaling constant $\varphi$ enters
as three structural inputs (T6--T8) established in the Recognition
Science (RS) framework of~\cite{Washburn2026Axioms,PardoGuerra2026}.  A
reader interested only in the finite-lattice results may take the model
as defined by this action and geometry, with the period $8$ and $D=3$
playing the role of boundary conditions of the kind standard in lattice
field theory.

\subsection{The composition law and the d'Alembert classification (T5)}

We adopt three hypotheses on $J:\mathbb R_{>0}\to\mathbb R$:
\begin{align}
  &\text{A1 (Normalization):} && J(1) = 0, \label{eq:A1}\\
  &\text{A2 (Composition):} && J(xy) + J(x/y) = 2J(x)J(y) + 2J(x) + 2J(y), \label{eq:A2}\\
  &\text{A3 (Calibration):} && J''_{\log}(0) = 1. \label{eq:A3}
\end{align}
We denote these hypotheses~\eqref{eq:A1}--\eqref{eq:A3}.
Here $J_{\log}(t):=J(e^t)$.  The list is not minimal:
\eqref{eq:A1} follows from~\eqref{eq:A2} and~\eqref{eq:A3}.  Setting
$x=y=1$ in~\eqref{eq:A2} gives $2J(1)=2J(1)^2+4J(1)$, so
$J(1)\in\{0,-1\}$; the branch $J(1)=-1$ is the trivial constant solution
$J\equiv-1$, whose log-curvature $J_{\log}''(0)=0$ is incompatible with
the calibration~\eqref{eq:A3}, leaving $J(1)=0$.  This dependence is
established in~\cite{Washburn2026Axioms}; we keep~\eqref{eq:A1} in the
list for readability.  We read~\eqref{eq:A3} in the regularizing
sense of~\cite{Washburn2026Axioms}: the log-curvature calibration is
assumed together with the local regularity of $J_{\log}$ near $0$ needed
to apply the continuous d'Alembert classification.  In particular, this
excludes the pathological non-measurable solutions admitted by the bare
functional equation~\eqref{eq:A2}, so~\eqref{eq:A3} is treated here as an
explicit regularity-and-calibration hypothesis rather than a tacit
smoothness assumption.  The substitution
$H(t):=1+J(e^t)$ turns
~\eqref{eq:A2} verbatim into the d'Alembert functional equation
$H(s+t)+H(s-t)=2H(s)H(t)$.  Its continuous solutions with $H(0)=1$ are
$H(t)=\cosh(\lambda t)$, $H(t)=\cos(\beta t)$, or $H\equiv1$
(Acz\'el~\cite{Aczel1966}); all are automatically even.  The
calibration~\eqref{eq:A3} gives $H''(0)=J''_{\log}(0)=1>0$, whose positive
sign excludes the oscillatory branch $\cos(\beta t)$ (with
$H''(0)=-\beta^2<0$) and the constant branch (with $H''(0)=0$), and whose
value $1$ fixes $\lambda=1$ in the surviving hyperbolic branch, giving
\begin{equation}\label{eq:J-unique}
  J(x)=\tfrac12(x+x^{-1})-1.
\end{equation}
We label this step T5: it is the d'Alembert classification
specialized to the calibrated cost, equivalent to the cosh addition
formula.  Reciprocal symmetry $J(x)=J(x^{-1})$ and continuity on
$\mathbb R_{>0}$ follow from~\eqref{eq:J-unique}.

\subsection[Immediate consequences of J]{Immediate consequences of \texorpdfstring{$J$}{J}}

Once~\eqref{eq:J-unique} is in place, the following are elementary
consequences, not independent postulates.  Properties~(J1)--(J4) are
analytic corollaries of the cost formula; the stable-record reading
below is an interpretive convention, flagged separately.
\begin{align}
  \text{(J1) Zero set:} &&
  J(x)=0 \iff x=1; \label{eq:J1}\\
  \text{(J2) Boundary:} &&
  J(x)\xrightarrow[x\to 0^+]{} \infty; \label{eq:J2}\\
  \text{(J3) Unique minimum:} &&
  x=1 \text{ is the unique global minimum of } J \text{ on } \mathbb R_{>0};
  \label{eq:J3}\\
  \text{(J4) Reciprocity:} &&
  J(x)=J(x^{-1}) \quad \text{for all } x>0. \label{eq:J4}
\end{align}
(J1) is~\eqref{eq:J1}; (J2)--(J4) are~\eqref{eq:J2}--\eqref{eq:J4}.
(J1)--(J3) follow from the AM--GM inequality applied to
$J(x)=\tfrac12(x+x^{-1})-1\ge 0$ with equality iff $x=1$.
(J4) is explicit in~\eqref{eq:J-unique}.
\begin{quote}
\noindent Interpretation (stable records).
An observable is a comparison ratio $x>0$ held near balance
($x\approx 1$); stable readout corresponds to small $J(x)$, and exact
balance is the zero-cost stratum~\eqref{eq:J1}.  This reading is
interpretive rather than a new analytic theorem.
\end{quote}

\subsection{Structural inputs (T6--T8)}

Three further inputs fix the constants of the discrete substrate.  Each
is derived in the RS framework of~\cite{Washburn2026Axioms,PardoGuerra2026}
and is taken here as given; the labels T6--T8 follow the theorem
numbering of the accompanying Lean formalization (Appendix~\ref{app:lean}):
\begin{enumerate}[label=(\arabic*), nosep]
  \item T6 ($\varphi$). The scaling ratio is the golden ratio
        $\varphi=(1+\sqrt5)/2$, the positive root of $r^2=r+1$.
  \item T7 (period $8$). The temporal clock is cyclic with
        period $8$, so the time direction is $\mathbb Z/8\mathbb Z$.
  \item T8 ($D=3$). The spatial lattice has dimension $D=3$, so
        the spatial directions are $\mathbb Z^3$.
\end{enumerate}
These are model-defining hypotheses, not consequences
of~\eqref{eq:J-unique}.  Together with the bond action
$\cosh(\Delta\phi)-1$ fixed by T5, the period $8$ (T7) and the dimension
$D=3$ (T8) determine the lattice $\mathbb Z^3\times\mathbb Z/8\mathbb Z$
on which the rest of the paper works; they enter as fixed boundary data,
as in any lattice field theory.  A Lean formalization threading the
T0--T8 chain is recorded in Appendix~\ref{app:lean}.

On the 8-mode signal space $\mathbb C^8$,
the \emph{recognition operator} $\hat{R}$ is the one-tick cyclic shift
$U$ carried to the DFT-8 basis (Proposition~\ref{prop:dft8}): for mode
amplitudes $\psi_k$,
\[
  (\hat R\psi)_k = \omega^{-k}\psi_k,
  \qquad \omega=e^{-2\pi i/8}.
\]
This is the linear update used in Appendix~\ref{sec:schrodinger}; a
cost-minimization formulation of the same step is given
in~\cite{Washburn2026Axioms}.  Everything in
Sections~\ref{sec:complex}--\ref{sec:rp-transfer} uses only T5, the
period $8$ (T7), and the finite lattice from T8; T6
($\varphi$) enters only in Appendix~\ref{sec:spacetime}, where we record
the value $J(\varphi)$.

\subsection[Constants]{Constants}

The only constant entering the finite-lattice statements below is the
causal speed $c=\text{l}_0/\tau_0=1$ in the units of
Section~\ref{sec:notation}.  The action-to-phase constant
$\hbar_{\rm RS}$ of $e^{iS/\hbar_{\rm RS}}$ enters no theorem of this
paper and may be set to $1$.

\section{Cyclic Shift and Action Identity}
\label{sec:complex}

This section records two elementary facts used below: the spectral
representation of the cyclic shift, and the identity
$J(e^\varepsilon)=\cosh\varepsilon-1$ that identifies the reciprocal cost
with a Euclidean action density.

The first is a period-agnostic fact of linear algebra, stated
below for a general period $n$ and then specialized to the
period $n=8$: a real cyclic shift of period $n>2$ has no complete real
one-dimensional eigendecomposition, and over $\mathbb C$ it diagonalizes
via the discrete Fourier transform.

\subsection{The Shift Operator}

The period-8 clock defines a cyclic shift operator $U$ on real signals
$f : \{0,1,\ldots,7\} \to \mathbb{R}$:
\begin{equation}
  (Uf)(k) = f((k + 1) \bmod 8).
\end{equation}
The signal space is $\mathbb{R}^8$ at this stage; iterating the index
advance eight times gives $U^8=\mathrm{id}$, and the complex spectral
representation is derived below, not assumed.

\subsection{Complex spectrum and DFT diagonalization}

Viewed as a real matrix, the period-$n$ shift $U_n$ is a permutation with
characteristic polynomial $\chi_{U_n}(\lambda)=\lambda^n-1$, whose roots
are the $n$th roots of unity.

\begin{proposition}[Complex spectral representation of the cyclic shift]\label{thm:complex_spectral}
  For every $n>2$ the cyclic shift $U_n$ on $\mathbb R^n$ has no complete
  real one-dimensional eigendecomposition: $\lambda^n-1$ has non-real
  roots, so over $\mathbb R$ the corresponding modes are carried by
  invariant two-dimensional rotation blocks, while over $\mathbb C$ the
  $n$ distinct roots diagonalize $U_n$.  For $n=8$,
  \[
    \lambda^8-1=(\lambda-1)(\lambda+1)(\lambda^2+1)
      (\lambda^2-\sqrt2\,\lambda+1)(\lambda^2+\sqrt2\,\lambda+1),
  \]
  and the irreducible factor $\lambda^2+1=(\lambda-i)(\lambda+i)$ shows
  that $\pm i$ are eigenvalues of $U=U_8$.
\end{proposition}
\begin{proof}
  The roots of $\lambda^n-1$ are $e^{2\pi i m/n}$, $m=0,\ldots,n-1$.  For
  $n>2$ at least one root is non-real, so $\chi_{U_n}$ does not split into
  real linear factors and $U_n$ has no real one-dimensional
  eigendecomposition; each conjugate pair $e^{\pm 2\pi i m/n}$ is carried
  over $\mathbb R$ by a two-dimensional rotation block, and the $n$
  distinct roots give a diagonalization over $\mathbb C$.  For $n=8$ the
  factor $\lambda^2+1$ is irreducible over $\mathbb R$ (since $x^2+1>0$)
  and equals $(\lambda-i)(\lambda+i)$, so $\pm i$ are eigenvalues of $U$.
\end{proof}

The value $8$ thus enters only as the period and is irrelevant to
the linear algebra.  The order-$n$ discrete Fourier transform
diagonalizes $U_n$ explicitly and is unitary.

\begin{proposition}[DFT diagonalization]\label{prop:dft8}
  Let $\omega_n=e^{-2\pi i/n}$ and define the order-$n$ DFT by
  \[
    (\widehat f)_k=\frac{1}{\sqrt n}\sum_{j=0}^{n-1} f(j)\,\omega_n^{jk}.
  \]
  Then $\widehat{U_n f}=D\widehat f$ with $D_{kk}=\omega_n^{-k}$, and
  $\langle \widehat f,\widehat g\rangle=\langle f,g\rangle$ for all
  $f,g\in\mathbb C^n$.  For $n=8$ this is the DFT-8 used below, with
  $\omega:=\omega_8=e^{-2\pi i/8}$.
\end{proposition}
\begin{proof}
  A change of variables $m=j+1$ in the cyclic sum gives
  \[
    (\widehat{U_n f})_k
    =\frac{1}{\sqrt n}\sum_{j=0}^{n-1} f(j+1)\,\omega_n^{jk}
    =\omega_n^{-k}\,\frac{1}{\sqrt n}\sum_{m=0}^{n-1} f(m)\,\omega_n^{mk}
    =\omega_n^{-k}(\widehat f)_k .
  \]
  With $F_{kj}=n^{-1/2}\omega_n^{jk}$ and conjugate transpose $F^\ast$,
  \[
    (F F^\ast)_{k\ell}
    =\frac{1}{n}\sum_{j=0}^{n-1}\omega_n^{j(k-\ell)}
    =\delta_{k\ell},
  \]
  since the geometric sum vanishes unless $k\equiv\ell\pmod n$.  Thus $F$
  is unitary, and Parseval holds for the standard Hermitian inner product
  $\langle f,g\rangle=\sum_{j=0}^{n-1}\overline{f(j)}\,g(j)$.
\end{proof}

\subsection[J-Cost as Euclidean Action]{\texorpdfstring{$J$}{J}-Cost as Euclidean Action}
\label{sec:wick}

$J$-cost equals the Euclidean action density, and a formal Wick
correspondence follows.  The rigorous Osterwalder--Schrader bridge to a
Lorentzian theory requires reflection positivity and is not established
unconditionally here.

\begin{lemma}[Euclidean action identity]\label{thm:euclidean}
  For all $\varepsilon \in \mathbb{R}$,
  \[
    J(e^\varepsilon) = \cosh(\varepsilon) - 1.
  \]
\end{lemma}
\begin{proof}
  By definition, $J(x) = \frac{1}{2}(x + x^{-1}) - 1$.  Setting $x = e^\varepsilon$:
  $J(e^\varepsilon) = \frac{1}{2}(e^\varepsilon + e^{-\varepsilon}) - 1 = \cosh(\varepsilon) - 1$.
\end{proof}

The leading-order term $\varepsilon^2/2$ is the standard free-scalar
kinetic-action term on the lattice; the full $\cosh\varepsilon-1$ is a
specific interacting action with all even-power vertices at fixed
coefficients $1/(2n)!$.  In the small-perturbation regime:
\begin{lemma}[Quadratic approximation]\label{thm:quadratic}
  For $|\varepsilon| \leq 1/2$,
  \[
    \left|J(e^\varepsilon) - \frac{\varepsilon^2}{2}\right| \leq \frac{|\varepsilon|^4}{18}.
  \]
  (This bound is used in the continuum-scaling discussion of
  Section~\ref{sec:continuum-limit}.)
\end{lemma}
\begin{proof}
  Write
  \[
    \cosh(\varepsilon)-1-\frac{\varepsilon^2}{2}
    =\sum_{n=2}^{\infty}\frac{\varepsilon^{2n}}{(2n)!}.
  \]
  For $|\varepsilon|\le 1/2$, the tail satisfies $24/(2k+4)!\le 1$ for all
  $k\ge 0$ (with equality at $k=0$), so
  \[
    \sum_{n=2}^{\infty}\frac{|\varepsilon|^{2n}}{(2n)!}
    =\frac{|\varepsilon|^4}{24}
      \sum_{k=0}^{\infty}|\varepsilon|^{2k}\frac{24}{(2k+4)!}
    \le \frac{|\varepsilon|^4}{24}\sum_{k=0}^{\infty}|\varepsilon|^{2k}
    =\frac{|\varepsilon|^4}{24}\cdot\frac{1}{1-|\varepsilon|^2}
    \le \frac{|\varepsilon|^4}{24}\cdot\frac{4}{3}
    =\frac{|\varepsilon|^4}{18},
  \]
  where the last bound uses $|\varepsilon|^2\le 1/4$.
\end{proof}

Lemma~\ref{thm:quadratic} controls the deviation from the Gaussian
kinetic term.

The leading term $\varepsilon^2/2$ gives the free-scalar lattice
kinetic term under the discrete gradient $\Delta_\mu\phi$.  The
$\varepsilon^4/24$ term is a quartic gradient interaction, namely
$(\Delta_\mu\phi)^4/24$ on a bond, not a local scalar $\phi^4$
potential.  Higher even gradient vertices are likewise fixed by the cosh
expansion, with no free parameters.

\begin{remark}[Formal Lorentzian interpretation]
The identity $\cosh(\tau)=\cos(i\tau)$ formally connects the Euclidean and
Lorentzian signatures: continuing $\tau\to it$ replaces the decaying
Euclidean bond weight $e^{-(\cosh(\Delta_\mu\phi)-1)}$ by a Lorentzian
phase factor $e^{iS/\hbar_{\rm RS}}$ of unit modulus.  This is only a
formal change of variables, not a reconstruction of a quantum theory: a
genuine passage to a Lorentzian theory would proceed through
Osterwalder--Schrader reconstruction and hence requires reflection
positivity, which the continuous noncompact kernel does not supply
(Proposition~\ref{prop:CB-fails}).  The discrete-field
Theorem~\ref{thm:rp-discrete} provides the OS hypothesis under
$\CBcond{v_0}{N}$, but the reconstruction is not carried out here.
\end{remark}

\section{Reflection Positivity and the Transfer Matrix}
\label{sec:rp-transfer}

To connect to constructive QFT one needs reflection positivity for the
Euclidean theory.  Section~\ref{subsec:lattice-action} fixes the
reciprocal-cost action and the single-site pinned finite-volume measure;
the remaining subsections develop reflection positivity along a single
logical spine, which we state here so that the formal results below read
as one chain rather than a list.
\begin{enumerate}[label=(\arabic*),nosep]
  \item By the standard crossing-bond factorization, reflection
        positivity for this nearest-neighbor action reduces to positive
        definiteness of a single one-bond temporal kernel
        $K(u)=\exp[-(\cosh u-1)]$ (Section~\ref{subsec:halfspace}).
  \item For the continuous noncompact field this reduction fails:
        $K$ is not positive definite, since its Fourier transform dips
        below zero near $\xi\approx3$ (Section~\ref{subsec:CB}).
  \item Restricting the field to a finite symmetric alphabet
        $\Phi(v_0,N)$ replaces the Bochner test on $\mathbb R$ by
        positive semidefiniteness of a finite crossing-bond Toeplitz
        matrix, which a diagonal-dominance certificate discharges
        uniformly in the alphabet size $N$ at explicit spacings
        (Section~\ref{subsec:discrete-rescue}).
  \item The verified condition then yields reflection positivity of the
        discrete-field measure (Theorem~\ref{thm:rp-discrete}) and, in
        turn, a positive, self-adjoint one-step transfer operator in an
        explicit reflection-positivity inner product
        (Sections~\ref{subsec:transfer-disc}--\ref{subsec:transfer}).
\end{enumerate}
Steps~(2) and~(3) are the paired negative and positive contribution of
the paper.  The arguments throughout are stated entirely in terms of the
finite periodic lattice $\mathbb Z^3\times\mathbb Z/8\mathbb Z$, the bond
potential $\cosh(\Delta\phi)-1$, and, in the discrete-field variant, a
finite admissible set $\Phi$; the RS origin of these inputs is noted but
not required to follow the mathematical content.

\subsection{The reciprocal-cost Euclidean lattice action and pinned measure}
\label{subsec:lattice-action}

Let the Euclidean lattice be
$\Lambda_{\rm RS}:=\mathbb Z^3\times \mathbb Z/8\mathbb Z$.  A real
scalar field is a function $\phi:\Lambda_{\rm RS}\to\mathbb R$.  For a
lattice direction $\hat\mu$, define the finite difference
$\Delta_\mu\phi(x):=\phi(x+\hat\mu)-\phi(x)$.
The nearest-neighbor action density is fixed by the reciprocal cost:
\[
J(e^\varepsilon)=\cosh\varepsilon-1.
\]
Formally, on the infinite lattice this gives
\begin{equation}
\label{eq:cosh-action}
S_{\rm RS}[\phi]
:=\sum_{x\in\Lambda_{\rm RS}}\sum_\mu
\left(\cosh(\Delta_\mu\phi(x))-1\right),
\end{equation}
where the inner sum runs over the four positive coordinate directions
$\hat\mu\in\{\hat e_1,\hat e_2,\hat e_3,\hat e_0\}$ (three spatial and one
temporal), so that the pair $(x,\hat\mu)$ indexes each
nearest-neighbor bond exactly once and no bond is double-counted.

For even $L$, set
\[
B_L:=\{-L/2,\ldots,L/2-1\}^3\subset\mathbb Z^3,
\qquad
\Lambda_L:=B_L\times\mathbb Z/8\mathbb Z.
\]
The temporal direction is periodic with period eight.  Spatial bonds may
be taken with periodic boundary conditions.  Alternatively, free boundary
conditions may be imposed by restricting the sum to nearest-neighbor
bonds entirely contained in $B_L$.  The reflection-positivity and
transfer-kernel statements below do not depend on this choice.  The
action depends only on field gradients, so it is
invariant under the single global additive shift $\phi\mapsto\phi+c$
for $c\in\mathbb R$.  This is the only gauge freedom of the action: a
temporally varying shift $\phi(x,t)\mapsto\phi(x,t)+c(t)$ alters the
temporal bonds $\cosh(\phi(x,t+1)+c(t+1)-\phi(x,t)-c(t))-1$ unless
$c(t)$ is constant, so per-slice shifts are not symmetries.

For the continuous measure we remove the one global additive zero
mode by a single global pinning: fix a reference site
$z_\ast=(x_\ast,t_\ast)\in\Lambda_L$ and impose
\begin{equation}
\label{eq:single-pin}
\phi(z_\ast)=0.
\end{equation}
This pin is a normalization device for the noncompact continuous measure
only.  The discrete-field measure of
Section~\ref{subsec:discrete-rescue}, by contrast, is supported on a
finite configuration space and is normalizable without any pin;
the reflection-positivity theorem there uses the unpinned finite measure,
so no $\theta$-compatibility constraint on $z_\ast$ arises.  The
finite-volume action is
\begin{equation}
\label{eq:finite-action}
S_L[\phi]
:=\sum_{z\in\Lambda_L}\sum_\mu
\left(\cosh(\Delta_\mu\phi(z))-1\right),
\end{equation}
with $\sum_\mu$ again running over the positive coordinate directions as
in~\eqref{eq:cosh-action} (one representative per nearest-neighbor bond),
and the pinned finite-volume measure is
\begin{equation}
\label{eq:pinned-measure}
d\mu_L(\phi)
:=Z_L^{-1}\,e^{-S_L[\phi]}
\prod_{z\in\Lambda_L\setminus\{z_\ast\}}d\phi(z),
\qquad \phi(z_\ast)=0,
\end{equation}
with $S_L$ as in~\eqref{eq:finite-action}.  The next lemma records that the
single global pin makes this noncompact continuous measure normalizable.

\begin{lemma}[Finiteness of the pinned partition function]
\label{lem:pinned-finite}
For every even $L$ the single-site pinned partition function
\[
  Z_L=\int_{\mathbb R^{\,\Lambda_L\setminus\{z_\ast\}}}
       e^{-S_L[\phi]}\!\!\!\prod_{z\in\Lambda_L\setminus\{z_\ast\}}\!\!\!d\phi(z),
  \qquad \phi(z_\ast)=0,
\]
is finite and strictly positive.  Hence $d\mu_L$
in~\eqref{eq:pinned-measure} is a well-defined probability measure.
\end{lemma}

\begin{proof}
View $\Lambda_L$ as the finite nearest-neighbor graph carrying the bonds
of $S_L$; with the periodic temporal cycle and either spatial boundary
condition on $B_L$ it is connected.  Fix a spanning tree $T$ rooted at
$z_\ast$; it has $|\Lambda_L|-1$ edges.  Change variables from the
unpinned site values $\{\phi(z):z\neq z_\ast\}$ to the tree-edge
differences $\eta_e:=\Delta\phi$ across $e\in T$.  Because $\phi(z_\ast)=0$
and each $\phi(z)$ equals the sum of the $\eta_e$ along the unique tree
path from $z_\ast$ to $z$, this linear map is a bijection given by an
integer unimodular matrix, so its Jacobian is $\pm1$.  Every bond term
$\cosh(\Delta\phi)-1$ is nonnegative and the tree edges form a subset of
all bonds, hence $S_L[\phi]\ge\sum_{e\in T}(\cosh\eta_e-1)$.  Therefore
\[
  0<Z_L\le\prod_{e\in T}\int_{\mathbb R}e^{-(\cosh\eta-1)}\,d\eta
        =\bigl(2e\,K_0(1)\bigr)^{\,|\Lambda_L|-1}<\infty ,
\]
where $\int_{\mathbb R}e^{-\cosh\eta}\,d\eta=2K_0(1)$ in terms of the
modified Bessel function $K_0$~\cite{DLMF}; the single-bond integral
converges because $\cosh\eta-1=\tfrac12 e^{|\eta|}(1+o(1))$ as
$|\eta|\to\infty$, so $e^{-(\cosh\eta-1)}$ decays faster than any
exponential.  Strict positivity of $Z_L$ follows from positivity and
continuity of the integrand.
\end{proof}

\subsection{Temporal reflection, half-lattice, and crossing bonds}
\label{subsec:halfspace}

The Osterwalder--Schrader half-space data are specified as follows.  We
use a \emph{link} (bond) reflection, so that the two halves are coupled
by genuine crossing bonds; this is the setting in which reflection
positivity reduces to positive definiteness of the one-bond kernel
$K$.  (A \emph{site} reflection $t\mapsto-t\bmod 8$, fixing the
slices $t=0,4$, would instead leave the two open arcs $\{1,2,3\}$ and
$\{5,6,7\}$ coupled only through the shared fixed-slice fields; for a
nearest-neighbor action that case is reflection positive
unconditionally (condition on the two fixed slices and use
$\theta$-symmetry), so it does not see the kernel $K$ at all and is not
the relevant reduction here.)  On the cyclic time
$\mathbb Z/8\mathbb Z=\{0,1,\ldots,7\}$ define the temporal reflection
\[
\theta(\mathbf{x},t)=(\mathbf{x},\,(7-t)\bmod 8),
\]
an involution with reflection planes lying between the slice
pairs $\{3,4\}$ and $\{7,0\}$ (it has no fixed slice).  It acts by
$0\leftrightarrow7$, $1\leftrightarrow6$, $2\leftrightarrow5$,
$3\leftrightarrow4$, and is a symmetry of the action $S_L$.  Take the
\emph{positive-time half-lattice}
\begin{equation}
\label{eq:halfspace}
\Lambda_L^+:=B_L\times\{0,1,2,3\},
\end{equation}
and the negative-time half-lattice $\Lambda_L^-:=B_L\times\{4,5,6,7\}$,
so $\theta(\Lambda_L^+)=\Lambda_L^-$.  The temporal bonds split as
follows:
\begin{itemize}[nosep]
\item Within $\Lambda_L^+$: bonds $t=0\!\to\!1$, $1\!\to\!2$, $2\!\to\!3$;
\item Within $\Lambda_L^-$: bonds $t=4\!\to\!5$, $5\!\to\!6$, $6\!\to\!7$;
\item Crossing bonds (one across each plane): $t=3\!\to\!4$ and
      $t=7\!\to\!0$.
\end{itemize}
Each crossing bond straddles a reflection plane and is mapped to itself
by $\theta$; it contributes a single factor
$K\bigl(\phi^+(\mathbf x)-\phi^-(\mathbf x)\bigr)$ to the bond-product
decomposition of $e^{-S_L}$, where $\phi^+,\phi^-$ are the fields on the
two slices adjacent to that plane (namely $(\phi_3,\phi_4)$ across one
plane and $(\phi_0,\phi_7)$ across the other).  Spatial bonds lie within
a single slice and so belong entirely to one half.  This is the
half-lattice / crossing-bond data (Figure~\ref{fig:reflection}) used in
all reflection-positivity statements below; the positive-time
half-lattice is~\eqref{eq:halfspace}.

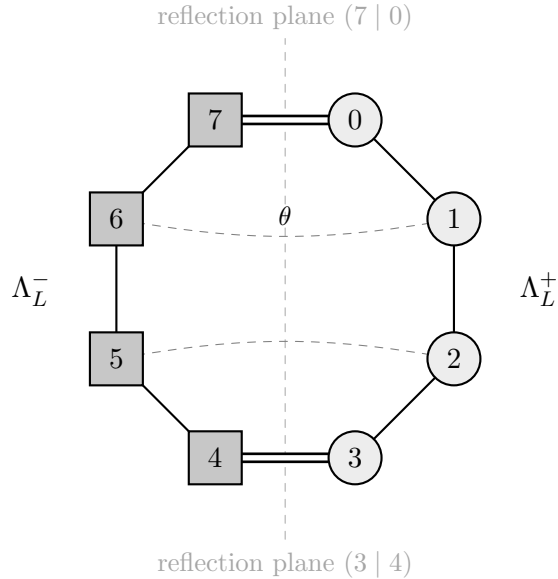
\begin{figure}[htbp]
\centering
\begin{tikzpicture}[scale=1.05,>=Stealth,line cap=round,
  posnode/.style={circle,draw=black,thick,fill=black!7,minimum size=7mm,inner sep=0pt,font=\small},
  negnode/.style={rectangle,draw=black,thick,fill=black!22,minimum size=7mm,inner sep=0pt,font=\small}]
  \def\R{2.3}
  \coordinate (p0) at (67.5:\R);
  \coordinate (p1) at (22.5:\R);
  \coordinate (p2) at (-22.5:\R);
  \coordinate (p3) at (-67.5:\R);
  \coordinate (p4) at (-112.5:\R);
  \coordinate (p5) at (-157.5:\R);
  \coordinate (p6) at (157.5:\R);
  \coordinate (p7) at (112.5:\R);
  \draw[dashed,gray!65] (0,-3.15) -- (0,3.15);
  \node[gray!70,font=\footnotesize] at (0,3.42) {reflection plane $(7\mid 0)$};
  \node[gray!70,font=\footnotesize] at (0,-3.46) {reflection plane $(3\mid 4)$};
  \draw[<->,gray,dashed,thin] (p1) to[bend left=10]
    node[pos=0.5,above,font=\footnotesize,text=black]{$\theta$} (p6);
  \draw[<->,gray,dashed,thin] (p2) to[bend right=10] (p5);
  \draw[thick] (p0)--(p1) (p1)--(p2) (p2)--(p3);
  \draw[thick] (p4)--(p5) (p5)--(p6) (p6)--(p7);
  \draw[line width=1pt,double,double distance=2pt,black] (p7)--(p0);
  \draw[line width=1pt,double,double distance=2pt,black] (p3)--(p4);
  \node[posnode] at (p0){0}; \node[posnode] at (p1){1};
  \node[posnode] at (p2){2}; \node[posnode] at (p3){3};
  \node[negnode] at (p4){4}; \node[negnode] at (p5){5};
  \node[negnode] at (p6){6}; \node[negnode] at (p7){7};
  \node[black,font=\small] at (3.2,0) {$\Lambda_L^+$};
  \node[black,font=\small] at (-3.2,0) {$\Lambda_L^-$};
\end{tikzpicture}
\caption{Temporal reflection on the 8-tick cycle
$\mathbb Z/8\mathbb Z$.  The involution
$\theta:(\mathbf x,t)\mapsto(\mathbf x,(7-t)\bmod 8)$ swaps
$0\!\leftrightarrow\!7$, $1\!\leftrightarrow\!6$,
$2\!\leftrightarrow\!5$, $3\!\leftrightarrow\!4$ (dashed identifications);
its two reflection planes lie in the gaps $7\mid 0$ and $3\mid 4$ (dashed
vertical axis).  The positive-time half-lattice
$\Lambda_L^+=B_L\times\{0,1,2,3\}$ (right, circular sites) and the
negative-time half $\Lambda_L^-=B_L\times\{4,5,6,7\}$ (left, square sites)
are joined only by the two
crossing bonds $7\!\to\!0$ and $3\!\to\!4$ (drawn as thick double bonds),
each straddling a reflection plane and contributing one factor of the
one-bond kernel $K$.}
\label{fig:reflection}
\end{figure}

Figure~\ref{fig:reflection} makes this bond bookkeeping visual.  Of the
eight temporal bonds, only the two thick double bonds cross a reflection
plane, while the remaining six temporal bonds (and all spatial bonds,
suppressed in the planar drawing) lie within a single half.  Reflecting
across the dashed vertical axis exchanges the right half $\Lambda_L^+$
with the left half $\Lambda_L^-$ and maps each crossing bond to itself,
so conditioning on the
four fields adjacent to the two planes factorizes $e^{-S_L}$ into a
half-supported amplitude, its $\theta$-image, and the two crossing-bond
weights.  Reflection positivity therefore reduces to a positivity property
of the single one-bond kernel $K$: the continuous cosh--Bochner test
of Section~\ref{subsec:CB} (Figure~\ref{fig:bochner}), and, once the field
alphabet is made finite, the matrix criterion $\CBcond{v_0}{N}$ certified
in Section~\ref{subsec:discrete-rescue} (Figure~\ref{fig:gershgorin}).  The
factorization just described is what the slice construction in the
proof of Theorem~\ref{thm:rp-discrete} carries out rigorously, assembling
the crossing-bond weights into the Gram matrix whose positivity those two
tests supply.

The reflection $\theta$ acts on the finite cyclic group
$\mathbb Z/8\mathbb Z$ and requires no special treatment of the periodic
boundary.  As noted in Section~\ref{subsec:lattice-action}, the
discrete-field measure used in Theorem~\ref{thm:rp-discrete} is finite
and is taken without a pin, so $\theta$ acts on the full
configuration space.

\subsection{The natural Bochner reduction and its failure}
\label{subsec:CB}

For a reflection-invariant finite-volume normalization, the standard
nearest-neighbor route to Osterwalder--Schrader reflection
positivity~\cite{OS1973,OS1975,FIS1978} is the inequality
\begin{equation}
\label{eq:rp-def}
\langle \theta F,F\rangle:=\int \overline{F(\theta\phi)}\,F(\phi)\,d\mu\ge 0
\end{equation}
for observables $F$ supported on the positive-time half-lattice
$\Lambda_L^+=B_L\times\{0,1,2,3\}$.  In the present nearest-neighbor
setup, the product/Gram argument reduces this condition to positive
definiteness of the single crossing-bond kernel.  Here $F$ may be complex-valued,
$\theta$ acts by the linear pullback $(\theta F)(\phi):=F(\theta\phi)$,
and $\langle G,H\rangle:=\int\overline{G}\,H\,d\mu$ is the Hermitian $L^2$
pairing, conjugate-linear in its first argument; the same convention is
used for the discrete-field measure $\mu_L^\Phi$ below
(Theorem~\ref{thm:rp-discrete}) and for the slice inner
product~\eqref{eq:RP-inner}.  We therefore test the one-bond crossing
kernel directly, before any reflection-compatible normalization of the
noncompact global zero mode is fixed: the single-site
pin~\eqref{eq:single-pin} is only a normalization device and is not
itself reflection-invariant for the link reflection.  For the $\cosh-1$
action this kernel is
\begin{equation}
\label{eq:K-onebond}
K(u):=\exp\!\bigl[-(\cosh u-1)\bigr],\qquad u\in\mathbb R,
\end{equation}
and, as we now show, its continuous Bochner positivity already fails.

\begin{definition}[Cosh--Bochner positivity hypothesis on $\mathbb R$]
\label{def:CB}
The \emph{cosh--Bochner positivity hypothesis} (CB) is the statement
that $K$ is positive definite on $\mathbb R$.  Since
$K\in L^1(\mathbb R)\cap C(\mathbb R)$, Bochner's theorem~\cite{Reed1975}
applies in both directions, so CB is equivalent to non-negativity of the
Fourier transform
\[
\widetilde K(\xi)
:=\int_{\mathbb R} K(u)\,e^{-i\xi u}\,du,
\qquad \widetilde K(\xi)\ge 0 \ \text{ for all } \xi\in\mathbb R.
\]
\end{definition}

The hypothesis CB is exactly what the standard nearest-neighbor route
requires; the next proposition evaluates $\widetilde K$ in closed form
and certifies that it turns negative near $\xi=3$.

\begin{proposition}[The continuous cosh kernel is not positive definite]
\label{prop:CB-fails}
The Fourier transform of $K(u)=\exp[-(\cosh u-1)]$ admits the closed form
$\widetilde K(\xi)=2\mathrm{e}\,K_{i\xi}(1)$, where $\mathrm{e}$ is
Euler's number and $K_{i\xi}$ is the modified Bessel function of imaginary
order~\cite{Watson1944,DLMF}.  Its value at
$\xi=3$ is negative,
\[
\widetilde K(3)\in[-0.004817599594096,\,-0.004817599594055]<0
\quad\text{(Figure~\ref{fig:bochner})},
\]
and $\widetilde K(\xi)<0$ on a nonempty open set of frequencies near
$\xi=3$.  Hence CB fails for the continuous noncompact kernel.
\end{proposition}
\begin{proof}
The closed form follows from the standard integral representation
\[
  K_{i\xi}(1)=\int_0^\infty \exp(-\cosh t)\cos(\xi t)\,dt
\]
(Watson~\cite{Watson1944}, \S6.22; equivalently DLMF~\cite{DLMF}, \S10.32.8).
Since $\exp[-(\cosh u-1)]=\mathrm{e}\cdot\exp(-\cosh u)$ and this kernel
is even in $u$,
\[
  \widetilde K(\xi)
  =\mathrm{e}\cdot\int_{\mathbb R}\exp(-\cosh u)\cos(\xi u)\,du
  =2\mathrm{e}\cdot\int_0^\infty \exp(-\cosh u)\cos(\xi u)\,du
  =2\mathrm{e}\,K_{i\xi}(1).
\]
The displayed enclosure of $\widetilde K(3)$ is the
interval certificate of Lemma~\ref{thm:fourier-negative}
(Appendix~\ref{app:fourier-cert}).  Since $K\in L^1(\mathbb R)$, its
Fourier transform $\widetilde K$ is continuous, so the strict inequality
$\widetilde K(3)<0$ persists on an open neighborhood of $\xi=3$.  The
discrete-field results below use only the finite matrix condition
$\CBcond{v_0}{N}$.
\end{proof}

\begin{figure}[htbp]
\centering
\includegraphics[width=0.82\linewidth]{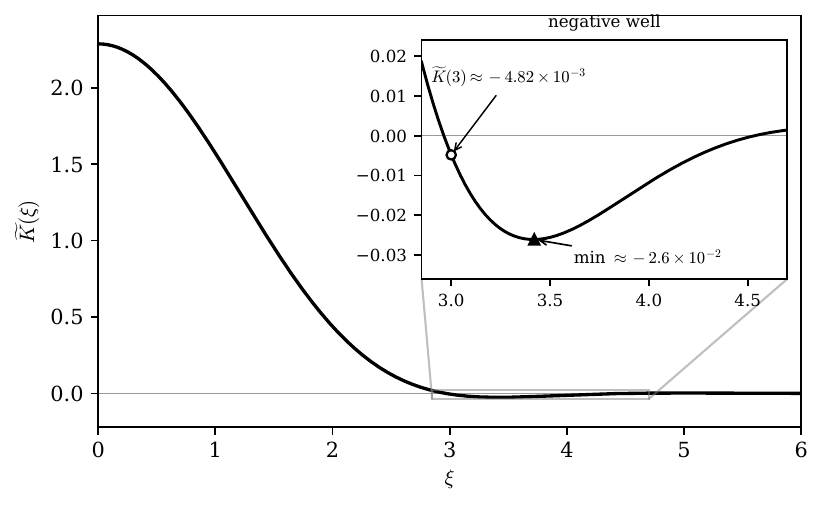}
\caption{Fourier transform
$\widetilde K(\xi)=\int_{\mathbb R}K(u)\,e^{-i\xi u}\,du$ of the one-bond
temporal kernel $K(u)=\exp[-(\cosh u-1)]$.  The transform is positive
near $\xi=0$ but becomes negative (inset) near $\xi\approx 3$, with the
rigorous enclosure $\widetilde K(3)\in[-4.818,\,-4.817]\times10^{-3}<0$
established in Lemma~\ref{thm:fourier-negative}.  By Bochner's theorem a negative
Fourier value shows that $K$ is not positive definite, so the
cosh--Bochner hypothesis (CB) fails for the continuous noncompact
kernel.}
\label{fig:bochner}
\end{figure}

Figure~\ref{fig:bochner} shows why this failure is genuine yet
quantitatively delicate.  The transform is dominated by a single positive
lobe of height $\widetilde K(0)=2\mathrm{e}\,K_0(1)\approx 2.29$, and its only
excursion below zero is one shallow, isolated well: $\widetilde K$ changes
sign near $\xi\approx 2.97$, reaches a minimum of about $-2.6\times10^{-2}$
near $\xi\approx 3.4$, and is positive again by $\xi\approx 4.6$.  The
certified value $\widetilde K(3)<0$ of
Lemma~\ref{thm:fourier-negative}, near $-4.8\times10^{-3}$, lies on the
near shoulder of this well (inset).  Because the dip is barely one percent of the central peak, the
violation cannot be read off by eye and must be established by the rigorous
enclosure of Appendix~\ref{app:fourier-cert}; equally, the smallness and
isolation of the well are precisely what let the discrete restriction of
Section~\ref{subsec:discrete-rescue} evade the obstruction, since the
periodized symbol~\eqref{eq:poisson} then samples a kernel whose narrow
negative region is negligible against the dominant positive mass near
$\xi=0$.

The failure of CB does not preclude reflection positivity of a
modified finite-field version of the model.  In that setting the
relevant question is no longer Bochner positivity on all of $\mathbb R$,
but positive semidefiniteness of a finite crossing-bond matrix.  This
passes from Fourier analysis on the noncompact group $\mathbb R$ to a
finite-dimensional Gram-matrix problem, equivalently the finite
harmonic-analysis question seen by the restricted difference set
$\Phi-\Phi$.  The remainder of this section records that criterion.

\subsection{The discrete-field finite-matrix criterion}
\label{subsec:discrete-rescue}

We restrict the admissible field values at each lattice site to
a finite, evenly spaced symmetric set
\begin{equation}
\label{eq:Phi-spacing}
\Phi:=v_0\,\{-N,-N+1,\ldots,N-1,N\}\subset\mathbb R,
\qquad v_0>0,\ N\in\mathbb N.
\end{equation}
We then need positive definiteness of $K$ only on the finite difference
set $\Phi-\Phi=v_0\,\{-2N,\ldots,2N\}$.

The restriction to a coarse alphabet acts as a high-frequency cutoff.
The continuous Bochner failure of Proposition~\ref{prop:CB-fails} is
confined to a narrow, shallow band of frequencies near $\xi\approx 3$
(Figure~\ref{fig:bochner}); sampling $K$ on the lattice $v_0\mathbb Z$
periodizes its Fourier transform, so that once the spacing $v_0$ is
large enough the dominant positive mass near $\xi=0$ outweighs that
isolated dip and the finite Toeplitz matrix can be positive.  The
periodization identity~\eqref{eq:poisson} below makes this quantitative;
the following definition fixes the finite criterion.

\begin{definition}[Discrete cosh--Bochner condition $\CBcond{v_0}{N}$]
\label{def:CB-disc}
For a spacing $v_0>0$ and an integer $N\ge 1$, write
$K_{\Phi(v_0,N)}\in\mathbb R^{(2N+1)\times(2N+1)}$ for the symmetric
Toeplitz matrix
\[
\bigl(K_{\Phi(v_0,N)}\bigr)_{j,k}\;:=\;K\!\bigl((j-k)\,v_0\bigr),
\qquad j,k\in\{-N,\ldots,N\},
\]
with $K$ as in~\eqref{eq:K-onebond}.  We say $\CBcond{v_0}{N}$ holds if
$K_{\Phi(v_0,N)}$ is positive semidefinite as a real symmetric matrix.
\end{definition}

$\CBcond{v_0}{N}$ is a finite-dimensional condition: for any specified
$v_0,N$ it can be verified by a direct eigenvalue computation, or by any
other rigorous certificate for positive semidefiniteness.  Its
Toeplitz/Herglotz extension on $\mathbb Z$ is controlled by the periodic
symbol
\begin{equation}
\label{eq:symbol}
f_{v_0}(\theta)\;:=\;\sum_{k\in\mathbb Z} K(k v_0)\,e^{-i k\theta},
\qquad \theta\in[-\pi,\pi].
\end{equation}
Here $\theta$ is the Fourier variable on $[-\pi,\pi]$ dual to the
lattice $\mathbb Z$ obtained after restricting $K$ to the spacing
$v_0\mathbb Z$.  Standard Poisson summation~\cite{Reed1975} identifies
$f_{v_0}$ as the periodization of $\widetilde K$:
\begin{equation}
\label{eq:poisson}
f_{v_0}(\theta)\;=\;\frac{1}{v_0}\sum_{n\in\mathbb Z}
\widetilde K\!\Bigl(\frac{\theta+2\pi n}{v_0}\Bigr).
\end{equation}
The summands sample $\widetilde K$ at frequency spacing $2\pi/v_0$, so
negativity of $\widetilde K$ at some frequencies does not by itself forbid
non-negativity of the symbol $f_{v_0}$ at coarse spacing.  Numerically the
symbol appears to become non-negative above a spacing near
\begin{equation}
\label{eq:v-star}
v_\ast\;\approx\;1.0604,
\end{equation}
with grid evaluation of~\eqref{eq:symbol} bracketing the threshold between
$v_{\rm fails}=1.060$ and $v_{\rm passes}=1.061$.  These computations are
evidence and a guide to finite checks; they are not used below as a
substitute for the finite hypothesis $\CBcond{v_0}{N}$.

\begin{remark}[Floating-point evidence for small finite $\Phi$]
\label{rem:CB-direct}
For any specific $(v_0,N)$, $\CBcond{v_0}{N}$ reduces to a finite
$(2N+1)\times(2N+1)$ eigenvalue check.  Double-precision computation
(Appendix~\ref{app:reproducibility}) corroborates the rigorous bounds
below: for the certified spacings $v_0\in\{1.2,1.5,2.5\}$ the smallest
eigenvalue stays well above zero and is essentially independent of $N$
(e.g.\ $\lambda_{\min}\approx 0.483$ at $v_0=1.5$ for all $N\le 200$),
whereas below the threshold $v_\ast$ some truncations are indefinite
(e.g.\ $\lambda_{\min}(v_0{=}1,N{=}4)\approx -8.2\times10^{-4}$),
consistent with the symbol becoming negative.  These are floating-point
figures (evidence, not a proof) which the next theorem replaces, for
$v_0\in\{1.2,1.5,2.5\}$, by a closed-form bound uniform in $N$.
\end{remark}

We now replace these numerics by a rigorous bound: $K(0)=1$, while the
off-diagonal entries decay fast enough that the symmetric Toeplitz matrix
is strictly diagonally dominant once the spacing is not too small.

\begin{theorem}[Rigorous diagonal-dominance certificate for $\CBcond{v_0}{N}$]
\label{thm:gershgorin}
For $v_0>0$ set
\[
s(v_0)\;:=\;2\sum_{k=1}^{\infty}\exp\!\bigl[-(\cosh(kv_0)-1)\bigr]
\;=\;2\sum_{k=1}^{\infty}K(kv_0),
\]
a convergent series.  By convexity of $k\mapsto\cosh(kv_0)$ it obeys the
closed-form bound
\begin{equation}
\label{eq:dd-bound}
s(v_0)\;\le\;\frac{2\,K(v_0)}{1-\rho(v_0)},
\qquad
\rho(v_0):=\exp\!\bigl[-(\cosh 2v_0-\cosh v_0)\bigr]\in(0,1).
\end{equation}
If $s(v_0)<1$, then for every $N\ge 1$ the matrix
$K_{\Phi(v_0,N)}$ of Definition~\ref{def:CB-disc} is positive definite,
with
\[
\lambda_{\min}\bigl(K_{\Phi(v_0,N)}\bigr)\;\ge\;1-s(v_0)\;>\;0
\qquad\text{uniformly in }N,
\]
so $\CBcond{v_0}{N}$ holds for all $N$.  The spacings
$v_0\in\{1.2,1.5,2.5\}$ are representative certified values in the
diagonally dominant regime $s(v_0)<1$ (Remark~\ref{rem:threshold}): they
are points at which the uniform-in-$N$ certificate is
discharged, not physically distinguished constants, and any
$v_0$ with $s(v_0)<1$ would serve equally well.
For the certified spacings used here, the closed-form
bound~\eqref{eq:dd-bound} gives the rigorous lower bounds
\[
\lambda_{\min}\ge 0.0894\ (v_0{=}1.2),\quad
\lambda_{\min}\ge 0.4825\ (v_0{=}1.5),\quad
\lambda_{\min}\ge 0.9882\ (v_0{=}2.5),
\]
each valid for all $N\ge 1$; see Figure~\ref{fig:gershgorin}.
\end{theorem}

\begin{figure}[htbp]
\centering
\includegraphics[width=0.82\linewidth]{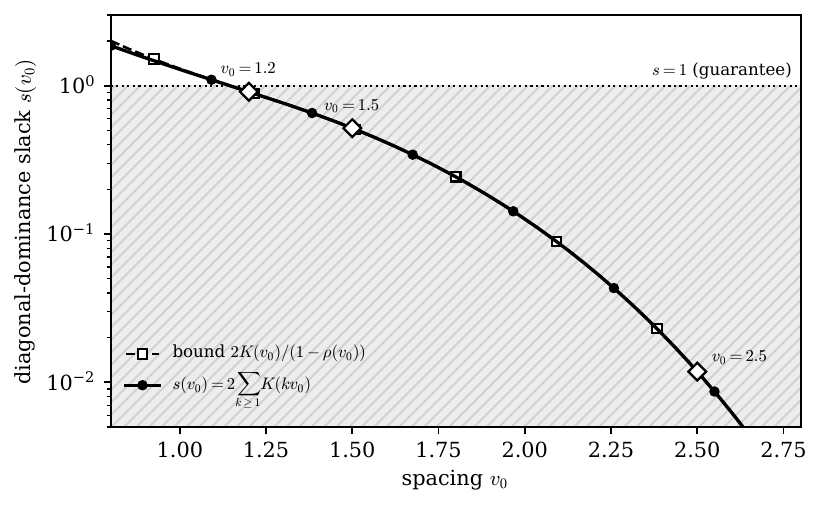}
\caption{Diagonal-dominance slack
$s(v_0)=2\sum_{k\ge 1}K(kv_0)$ (solid line, filled circles) and its
closed-form upper bound
$2K(v_0)/\bigl(1-\rho(v_0)\bigr)$ (dashed line, open squares), versus the
field spacing $v_0$, on a logarithmic scale.  Where $s(v_0)<1$ (hatched
region, below the dotted $s=1$ line), Theorem~\ref{thm:gershgorin}
guarantees
$\lambda_{\min}\bigl(K_{\Phi(v_0,N)}\bigr)\ge 1-s(v_0)>0$ for every $N$,
so $\CBcond{v_0}{N}$ holds uniformly in $N$.  The three certified
spacings $v_0\in\{1.2,1.5,2.5\}$ (open diamonds) give the rigorous bounds
$s\le 0.911,\,0.518,\,0.012$, respectively.}
\label{fig:gershgorin}
\end{figure}

\begin{proof}
Every diagonal entry of $K_{\Phi(v_0,N)}$ equals $K(0)=1$.  For a fixed
row index $j\in\{-N,\ldots,N\}$, the sum of the off-diagonal entries in
absolute value is
\[
\sum_{k\ne j}K\bigl((j-k)v_0\bigr)
=\sum_{\substack{m=j-k\\ m\ne 0}}K(mv_0)
\;\le\;\sum_{m\in\mathbb Z\setminus\{0\}}K(mv_0)
\;=\;2\sum_{m\ge 1}K(mv_0)\;=\;s(v_0),
\]
using $K>0$ and $K(-u)=K(u)$.  Hence if $s(v_0)<1$ the diagonal entry
$1$ strictly exceeds the off-diagonal absolute row sum, so
$K_{\Phi(v_0,N)}$ is a real symmetric strictly diagonally dominant
matrix with positive diagonal; by Gershgorin's theorem~\cite{HornJohnson}
its (real) eigenvalues lie in $[\,1-s(v_0),\,1+s(v_0)\,]$, giving
$\lambda_{\min}\ge 1-s(v_0)>0$ independently of $N$.
For~\eqref{eq:dd-bound}, the map $k\mapsto\cosh(kv_0)$ is convex, so for
$k\ge 1$ the secant slope from $1$ to $k$ dominates that from $1$ to
$2$, i.e.\ $\cosh(kv_0)\ge\cosh(v_0)+(k-1)\bigl(\cosh 2v_0-\cosh
v_0\bigr)$.  Therefore
$K(kv_0)\le K(v_0)\,\rho(v_0)^{\,k-1}$.  Since $\cosh$ is strictly
increasing on $[0,\infty)$ and $2v_0>v_0>0$, we have $\cosh 2v_0>\cosh
v_0$, so the ratio
$\rho(v_0)=\exp[-(\cosh 2v_0-\cosh v_0)]$ lies in $(0,1)$; summing the
geometric series then gives $\sum_{k\ge1}K(kv_0)\le K(v_0)/(1-\rho(v_0))$,
which is~\eqref{eq:dd-bound}.  Substituting $v_0\in\{1.2,1.5,2.5\}$ yields
$s\le 0.9106,\,0.5175,\,0.0118$ respectively, hence the stated bounds.
The quantities in~\eqref{eq:dd-bound} are explicit closed forms in
$\cosh$ and $\exp$; the inequality $s(v_0)<1$ is therefore decidable by
exact/interval evaluation and does not rely on floating-point
eigenvalues.
\end{proof}

The conclusion $\lambda_{\min}\bigl(K_{\Phi(v_0,N)}\bigr)\ge 1-s(v_0)>0$
is strict positive definiteness, which is stronger than the
positive semidefiniteness that Definition~\ref{def:CB-disc} requires for
$\CBcond{v_0}{N}$; at the certified spacings the condition therefore holds
in this stronger form, with an eigenvalue margin uniform in $N$.

Figure~\ref{fig:gershgorin} plots both the exact slack $s(v_0)$ and the
closed-form bound~\eqref{eq:dd-bound}, and two features are worth drawing
out.  First, for $v_0\gtrsim 1$ the two curves are visually
indistinguishable, so the rigorous bound sacrifices almost nothing relative
to the exact series; the decidable criterion $s(v_0)<1$ is essentially as
sharp as the diagonal-dominance argument permits.  Second, because
$s(v_0)$ inherits the super-exponential decay of
$K(kv_0)=\exp[-(\cosh(kv_0)-1)]$, the guaranteed eigenvalue floor
$1-s(v_0)$ widens extremely fast once the threshold is cleared: the slack
falls from $s\approx 0.91$ at $v_0=1.2$ to $s\approx 0.012$ at $v_0=2.5$.
These three certified spacings range from just inside the diagonally
dominant region to deep within it, matching the per-spacing eigenvalue
floors recorded in Theorem~\ref{thm:gershgorin}.  It is precisely this
$N$-independent
eigenvalue margin that discharges the finite-matrix hypothesis
$\CBcond{v_0}{N}$ uniformly in $N$ and so, through the
slice/Gram argument of Theorem~\ref{thm:rp-discrete}, delivers reflection
positivity of the discrete-field model at every certified spacing
(Corollary~\ref{cor:rp-verified}), the positive counterpart to the
continuous Bochner failure of Figure~\ref{fig:bochner}.

\begin{remark}[Numerical symbol threshold and the certified regime]
\label{rem:threshold}
The numerical symbol threshold $v_\ast\approx1.06$ of~\eqref{eq:v-star}
should be distinguished from the rigorous threshold of the
diagonal-dominance certificate: both the exact slack $s(v_0)$ and its
closed-form bound~\eqref{eq:dd-bound} cross the guarantee line $s=1$ near
$v_0\approx1.15$.  In the intervening range
$v_\ast\lesssim v_0\lesssim 1.15$ the periodized symbol appears
non-negative numerically, but the diagonal-dominance certificate does
not yet apply, so $\CBcond{v_0}{N}$ there remains a numerical observation
rather than a theorem.  The certified spacings $v_0\in\{1.2,1.5,2.5\}$
all lie above this range.
\end{remark}

Now define the \emph{discrete-field finite-volume measure}.  Each site
$z\in\Lambda_L$ carries a field $\phi(z)\in\Phi=\Phi(v_0,N)$.  Because
$\Phi$ is finite, the configuration space $\Phi^{|\Lambda_L|}$ is finite
and no pin is needed for normalizability.  The finite alphabet breaks the
continuous global-shift symmetry of the unrestricted gradient action,
rather than quotienting it.  Replace
the Lebesgue integration in~\eqref{eq:pinned-measure} by the counting
measure on $\Phi$ at every site:
\begin{equation}
\label{eq:pinned-measure-disc}
d\mu_L^{\Phi}(\phi)
:=Z_{L,\Phi}^{-1}\,e^{-S_L[\phi]}
\prod_{z\in\Lambda_L}d\nu_\Phi(\phi(z)),
\end{equation}
where $d\nu_\Phi$ is the uniform counting measure on $\Phi$ and
$Z_{L,\Phi}$ normalizes $\mu_L^\Phi$.  This is a finite measure on
$\Phi^{|\Lambda_L|}$; it is invariant under the reflection $\theta$ of
Section~\ref{subsec:halfspace} because $S_L$ is and the product is over
all sites.

\begin{definition}[Half-lattice field reflection]
\label{def:field-reflection}
Write $\phi_t$ for the spatial slice field at time $t$, i.e.\ the
restriction of a configuration $\phi$ to $B_L\times\{t\}$.  The temporal
reflection $\theta$ of Section~\ref{subsec:halfspace} induces the
\emph{half-lattice field reflection} $\Theta$, which sends $\phi$ to the
positive-half field
\[
  (\Theta\phi)_t:=\phi_{(7-t)\bmod 8}\qquad(t\in\{0,1,2,3\}),
\]
carrying the negative-half slices $\phi_4,\dots,\phi_7$ onto the positive
half-lattice $\Lambda_L^+$ slice-for-slice.  Equivalently, $\Theta\phi$
is the restriction to $\Lambda_L^+$ of the pulled-back configuration
$\theta\phi$.
\end{definition}

\begin{theorem}[Conditional discrete-field reflection positivity]
\label{thm:rp-discrete}
Fix $v_0>0$ and $N\ge 1$, and suppose $\CBcond{v_0}{N}$ holds for
$\Phi=\Phi(v_0,N)$.  Then the discrete-field finite-volume measure
$\mu_L^\Phi$ in~\eqref{eq:pinned-measure-disc} satisfies the
Osterwalder--Schrader reflection-positivity inequality
\[
\langle \theta F,F\rangle_{\mu_L^\Phi}\;\ge\; 0
\]
for every observable $F$ that is a function of the field on the
positive-time half-lattice $\Lambda_L^+=B_L\times\{0,1,2,3\}$.
\end{theorem}

\begin{proof}
This is the standard nearest-neighbor reflection-positivity argument
of~\cite{FIS1978}; we give it in full because the model is finite.  Order
the slices so that $\Lambda_L^+=B_L\times\{0,1,2,3\}$ and
$\Lambda_L^-=B_L\times\{4,5,6,7\}$, and recall (Section~\ref{subsec:halfspace})
that the only bonds joining the two halves are the two crossing temporal
bonds $t=3\!\to\!4$ and $t=7\!\to\!0$.  The Boltzmann weight factorizes
over bonds as
\[
e^{-S_L[\phi]}
= A(\phi_0,\phi_1,\phi_2,\phi_3)\;
  A(\phi_7,\phi_6,\phi_5,\phi_4)\;
  \prod_{\mathbf x\in B_L}
  K\!\bigl(\phi_3(\mathbf x)-\phi_4(\mathbf x)\bigr)\,
  K\!\bigl(\phi_0(\mathbf x)-\phi_7(\mathbf x)\bigr),
\]
where $A$ collects all bonds internal to a half (the three internal
temporal bonds and the spatial bonds of its four slices); the two
$A$-factors have the same functional form because $\theta$ maps
the negative half onto the positive half bond-for-bond.  In terms of the
half-lattice field reflection $\Theta$ of
Definition~\ref{def:field-reflection}, the negative-half factor is
$A\bigl((\Theta\phi)_0,\dots,(\Theta\phi)_3\bigr)$
and, since $K$ is even, each crossing factor couples a positive boundary
slice to its reflected partner:
$K(\phi_3-\phi_4)=K\bigl(\phi_3-(\Theta\phi)_3\bigr)$ and
$K(\phi_0-\phi_7)=K\bigl(\phi_0-(\Theta\phi)_0\bigr)$.

By hypothesis $\CBcond{v_0}{N}$, the $(2N+1)\times(2N+1)$ matrix
$\bigl(K(a-b)\bigr)_{a,b\in\Phi}$ is positive semidefinite, so it admits
a (real) factorization $K(a-b)=\sum_{\mu} u_\mu(a)\,u_\mu(b)$ with
finitely many vectors $u_\mu\in\mathbb R^\Phi$.  Applying this at each
site $\mathbf x\in B_L$ on each of the two reflection planes expands the
crossing product as a finite sum
\[
\prod_{\mathbf x}K\!\bigl(\phi_3(\mathbf x)-\phi_4(\mathbf x)\bigr)
K\!\bigl(\phi_0(\mathbf x)-\phi_7(\mathbf x)\bigr)
=\sum_{\alpha}
g_\alpha(\phi_0,\phi_3)\,g_\alpha\!\bigl((\Theta\phi)_0,(\Theta\phi)_3\bigr),
\]
where $\alpha$ ranges over the (finite) multi-index of the
$u_\mu$-factorizations over all boundary sites and each $g_\alpha$ is a
real function of the two positive-half boundary slices $\phi_0,\phi_3$.
A function $F$ on $\Lambda_L^+$ depends on $(\phi_0,\phi_1,\phi_2,\phi_3)$,
while $F(\theta\phi)$ depends on the reflected slices and equals
$F\bigl((\Theta\phi)_0,\dots,(\Theta\phi)_3\bigr)$.  Summing over the
negative-half configuration is, after the relabeling
$\phi_{4},\dots,\phi_{7}\leftrightarrow(\Theta\phi)_{3},\dots,(\Theta\phi)_0$,
the complex conjugate of the sum over the positive-half configuration,
since the weights $A$ and $g_\alpha$ are real-valued.  Hence
\[
\langle\theta F,F\rangle_{\mu_L^\Phi}
=Z_{L,\Phi}^{-1}\sum_{\alpha}\Bigl|\,
\sum_{\phi_0,\phi_1,\phi_2,\phi_3}
F(\phi_0,\dots,\phi_3)\,A(\phi_0,\dots,\phi_3)\,g_\alpha(\phi_0,\phi_3)
\Bigr|^{2}\;\ge\;0,
\]
a sum of squares, where each inner sum runs over $\Phi$-valued slice
fields.  This is the claimed inequality.  (Only positive
semidefiniteness of the one-bond matrix is used, not strict positivity;
the argument is insensitive to the global additive zero mode because no
pin is imposed.)
\end{proof}

Combining Theorem~\ref{thm:rp-discrete} with the rigorous
certificate of Theorem~\ref{thm:gershgorin} yields a reflection-positivity
statement for the discrete-field model at concrete spacings that does
not rest on floating-point eigenvalues and holds for all box truncations.

\begin{corollary}[Discrete-field reflection positivity at certified spacings, rigorous and uniform in $N$]
\label{cor:rp-verified}
For each spacing $v_0\in\{1.2,1.5,2.5\}$ and every $N\ge 1$, the
discrete-field finite-volume measure
$\mu_L^{\Phi(v_0,N)}$ in~\eqref{eq:pinned-measure-disc} satisfies the
Osterwalder--Schrader reflection-positivity inequality
$\langle\theta F,F\rangle_{\mu_L^{\Phi(v_0,N)}}\ge 0$ for every
observable $F$ supported on the positive-time half-lattice
$\Lambda_L^+=B_L\times\{0,1,2,3\}$.
\end{corollary}

\begin{proof}
For each $v_0\in\{1.2,1.5,2.5\}$ the closed-form
bound~\eqref{eq:dd-bound} gives $s(v_0)<1$, so
Theorem~\ref{thm:gershgorin} discharges $\CBcond{v_0}{N}$ rigorously and
uniformly in $N$ (with the displayed lower bounds on $\lambda_{\min}$).
Theorem~\ref{thm:rp-discrete} then applies for every $N\ge 1$.  The
certificate uses only the explicit inequality $s(v_0)<1$
in~\eqref{eq:dd-bound}, not the floating-point eigenvalues of
Remark~\ref{rem:CB-direct}.
\end{proof}

\begin{remark}[Scope of the discrete-field result]
\label{rem:rescue-scope}
The discrete-field model is a genuine modification of the action
functional (a finite-alphabet cutoff in the spirit of regularizations
standard in lattice field theory~\cite{Glimm1987}) and
Theorem~\ref{thm:rp-discrete} establishes OS reflection positivity for it
only after the finite matrix condition $\CBcond{v_0}{N}$ is verified.  At
the certified spacings $v_0\in\{1.2,1.5,2.5\}$ the result is
unconditional in the strict sense, following from the closed-form
inequality~\eqref{eq:dd-bound} and Theorem~\ref{thm:gershgorin}; at
small spacings some truncations are indefinite
(Remarks~\ref{rem:CB-direct} and~\ref{rem:threshold}), and at other
spacings the claim is numerical unless a separate diagonal-dominance,
exact $LDL^\top$, or interval-Cholesky certificate is supplied.  The result does not establish
OS positivity for the continuous noncompact measure $\mu_L$, nor that the
discrete-field model shares its continuum limit; continuum limits with
$v_0\to0$, $N\to\infty$ are a separate scaling problem
(Section~\ref{sec:continuum-limit}).
\end{remark}

\subsection{Discrete-field transfer operator and RP inner product}
\label{subsec:transfer-disc}

Fix $v_0>0$, $N\ge 1$ with $\CBcond{v_0}{N}$, and write
$\Phi=\Phi(v_0,N)$.  A \emph{slice configuration} is a map
$\phi:B_L\to\Phi$; write $\mathcal{X}_\Phi:=\Phi^{|B_L|}$ for the
finite set of such maps.  The spatial action on one slice and the
one-step temporal action between adjacent slices are
\[
S_{\rm sp}(\phi)
=\sum_{x\in B_L}\sum_{\hat\mu\ \text{spatial}}
\left(\cosh(\Delta_{\hat\mu}\phi(x))-1\right),
\qquad
S_{\rm temp}(\phi,\psi)
=\sum_{x\in B_L}\left(\cosh(\psi(x)-\phi(x))-1\right),
\]
with $\phi,\psi\in\mathcal{X}_\Phi$.  The corresponding
\emph{discrete one-step transfer kernel} is
\begin{equation}
\label{eq:transfer-kernel-disc}
\mathcal{T}^{\Phi}(\phi,\psi)
:=\exp\!\left(-\tfrac12 S_{\rm sp}(\phi)\right)\,
 \exp\!\left(-S_{\rm temp}(\phi,\psi)\right)\,
 \exp\!\left(-\tfrac12 S_{\rm sp}(\psi)\right),
\qquad \phi,\psi\in\mathcal{X}_\Phi .
\end{equation}
This kernel is defined directly by~\eqref{eq:transfer-kernel-disc}; it is
strictly positive (exponential of a real number) and symmetric in
$(\phi,\psi)$ because $\cosh$ is even.  The continuous-field kernel
\eqref{eq:transfer-kernel} of Section~\ref{subsec:transfer} is the same
expression with $\Phi$-valued slice fields replaced by real-valued ones;
Lemma~\ref{thm:transfer} records the corresponding pointwise statement
there.

Equip $\mathbb{C}^{\mathcal{X}_\Phi}$ with the
\emph{RP slice inner product} weighted by the spatial half-measure
\begin{equation}
\label{eq:RP-inner}
\langle f,g\rangle_{\rm RP}
:=\sum_{\phi\in\mathcal{X}_\Phi}
\overline{f(\phi)}\,g(\phi)\,
e^{-S_{\rm sp}(\phi)},
\qquad f,g:\mathcal{X}_\Phi\to\mathbb{C}.
\end{equation}
This Hermitian form (conjugate-linear in $f$ and linear in $g$, with the
same conjugation convention as the Osterwalder--Schrader pairing
in~\eqref{eq:rp-def}) is the inner product induced on single-slice
observables by the
Gibbs weight $e^{-S_L[\phi]}$ when the temporal bonds to neighboring
slices are treated as external data; it is the finite-dimensional
specialization of the Osterwalder--Schrader / transfer-matrix inner
product used in~\cite{FIS1978,Glimm1987,OsterwalderSeiler1978}.
Write $w(\phi):=e^{-S_{\rm sp}(\phi)/2}$ and
$K(\phi,\psi):=\exp[-S_{\rm temp}(\phi,\psi)]$, so
$\mathcal{T}^{\Phi}(\phi,\psi)=w(\phi)\,K(\phi,\psi)\,w(\psi)$.
Splitting the spatial action symmetrically between the two slices in
this way is exactly what makes the transfer step self-adjoint in
$\langle\cdot,\cdot\rangle_{\rm RP}$: since $K(\phi,\psi)=K(\psi,\phi)$,
the spatial weight $e^{-S_{\rm sp}}$ carried by~\eqref{eq:RP-inner}
absorbs the asymmetric ratio $w(\phi)/w(\psi)$ appearing in $T_\Phi$
below, leaving the symmetric kernel $w(\phi)\,K(\phi,\psi)\,w(\psi)$.
Define the \emph{one-step transfer operator}
$T_\Phi:\mathbb{C}^{\mathcal{X}_\Phi}\to\mathbb{C}^{\mathcal{X}_\Phi}$ by
\begin{equation}
\label{eq:transfer-op}
(T_\Phi f)(\psi)
:=\sum_{\phi\in\mathcal{X}_\Phi}
K(\phi,\psi)\,\frac{w(\phi)}{w(\psi)}\,f(\phi),
\qquad \psi\in\mathcal{X}_\Phi .
\end{equation}
This is the transfer step that is self-adjoint and positive in
$\langle\cdot,\cdot\rangle_{\rm RP}$; the full positive kernel
$\mathcal{T}^{\Phi}(\phi,\psi)$ in~\eqref{eq:transfer-kernel-disc} is
the symmetrized matrix element
$w(\psi)\,(T_\Phi)_{\psi,\phi}\,w(\phi)$ in the slice basis.

\begin{corollary}[Positive finite-box transfer matrix in the RP inner product]
\label{cor:transfer-psd}
Suppose $\CBcond{v_0}{N}$ holds for $\Phi=\Phi(v_0,N)$ (so that, by
Theorem~\ref{thm:rp-discrete}, the finite-box discrete-field measure
$\mu_L^\Phi$ is reflection positive).  Then the
one-step transfer matrix on
$\mathcal{X}_\Phi=\Phi^{|B_L|}$, equivalently the operator $T_\Phi$
in~\eqref{eq:transfer-op}, is positive self-adjoint in the explicit RP
inner product~\eqref{eq:RP-inner}:
\[
\langle f,T_\Phi f\rangle_{\rm RP}\;\ge\; 0
\qquad\text{for all }f:\mathcal{X}_\Phi\to\mathbb{C},
\]
and $\langle f,T_\Phi g\rangle_{\rm RP}=\langle T_\Phi f,g\rangle_{\rm RP}$
for all $f,g$.  Equivalently, the symmetric matrix
$\bigl(\mathcal{T}^{\Phi}(\phi,\psi)\bigr)_{\phi,\psi\in\mathcal{X}_\Phi}$
is positive semidefinite.
\end{corollary}

\begin{proof}
Self-adjointness in $\langle\cdot,\cdot\rangle_{\rm RP}$ follows from
symmetry of $K$ and a short calculation:
$\langle f,T_\Phi g\rangle_{\rm RP}=\langle T_\Phi f,g\rangle_{\rm RP}$.
For positivity, set $h(\phi):=w(\phi)f(\phi)$.  Then
\[
\langle f,T_\Phi f\rangle_{\rm RP}
=\sum_{\psi,\phi}\overline{f(\psi)}\,w(\psi)^2\,
 K(\phi,\psi)\,\frac{w(\phi)}{w(\psi)}\,f(\phi)
=\sum_{\phi,\psi}\overline{h(\psi)}\,K(\phi,\psi)\,h(\phi).
\]
Because the temporal action is a sum over sites, the slice kernel
factorizes over the spatial box,
\[
K(\phi,\psi)=\exp[-S_{\rm temp}(\phi,\psi)]
=\prod_{x\in B_L}K\bigl(\phi(x)-\psi(x)\bigr),
\]
so on the factored configuration space $\mathcal{X}_\Phi=\Phi^{B_L}$ the
slice matrix is a Kronecker product of one-bond matrices, one factor per
spatial site,
\[
\bigl(K(\phi,\psi)\bigr)_{\phi,\psi\in\mathcal{X}_\Phi}
=\bigotimes_{x\in B_L}\bigl(K(a-b)\bigr)_{a,b\in\Phi}.
\]
Each factor $(K(a-b))_{a,b\in\Phi}$ is positive semidefinite by
$\CBcond{v_0}{N}$, and a Kronecker product of positive semidefinite
matrices is positive semidefinite, so the displayed quadratic form is
nonnegative.  This is the standard
transfer-matrix positivity step for reflection-positive nearest-neighbor
models in the Osterwalder--Schrader/Fr\"ohlich--Israel--Lieb--Simon framework
\cite{OS1973,OS1975,FIS1978,Glimm1987}; Theorem~\ref{thm:rp-discrete}
supplies the finite-box half-space reflection-positivity input, and
$\CBcond{v_0}{N}$ supplies the one-bond positive definiteness used at
each spatial site.
Positivity of the symmetrized kernel
$\mathcal{T}^{\Phi}(\phi,\psi)=w(\phi)K(\phi,\psi)w(\psi)$ follows by
the same change of variables $h(\phi)=w(\phi)f(\phi)$.  The argument is
finite-box and finite-alphabet; no thermodynamic limit $L\to\infty$ is
taken.
\end{proof}

Corollary~\ref{cor:transfer-psd} is a finite-box, finite-alphabet
statement: on each spatial box $B_L$ the operator $T_\Phi$ is a positive
self-adjoint endomorphism of the RP slice space
$(\mathbb{C}^{\mathcal{X}_\Phi},\langle\cdot,\cdot\rangle_{\rm RP})$, and
composing it along the three internal temporal bonds of the positive
half-lattice gives the finite-volume transfer-matrix product underlying
Theorem~\ref{thm:rp-discrete}.  Spectral gaps, the limit $L\to\infty$,
and passage to a continuum theory are deferred (Section~\ref{sec:outlook}).

\subsection{Finite-box transfer kernel}
\label{subsec:transfer}

Because the action~\eqref{eq:cosh-action} depends only on field
gradients, configurations related by a global additive shift
$\phi\mapsto\phi+c$ have the same action.  For transfer-kernel
statements we use the same single global pin~\eqref{eq:single-pin} as in
the continuous pinned measure: $\phi(z_\ast)=0$ at a single reference
site $z_\ast=(x_\ast,t_\ast)$.  Slices not containing
$z_\ast$ are unpinned spatial-field configurations.  The one-step
transfer kernel below is therefore defined on adjacent slice pairs
$(\phi,\psi)$ modulo a common additive constant,
$(\phi,\psi)\sim(\phi+c,\psi+c)$.  This quotient is the slice-level
version of the same global zero-mode removal; it is not a per-slice
pinning and does not impose $\phi(x_\ast)=\psi(x_\ast)=0$ at every time.

For the finite spatial box $B_L=\{-L/2,\ldots,L/2-1\}^3\subset\mathbb Z^3$
with $x_\ast\in B_L$, let $\phi,\psi:B_L\to\mathbb R$ denote adjacent
temporal slices, considered as an ordered pair modulo the common shift
above.  One may choose a representative such as $\phi(x_\ast)=0$, leaving
$\psi(x_\ast)$ free as the relative temporal value.  Spatial bonds are
taken with periodic boundary conditions, or equivalently with the sum
restricted to nearest-neighbor bonds entirely contained in $B_L$; either
convention yields the same kernel statements below.  Define the spatial
action on one slice by
\[
S_{\rm sp}(\phi)
=\sum_{x\in B_L}\sum_{\hat\mu\ \text{spatial}}
\left(\cosh(\Delta_{\hat\mu}\phi(x))-1\right),
\]
the one-step temporal action by
\[
S_{\rm temp}(\phi,\psi)
=\sum_{x\in B_L}\left(\cosh(\psi(x)-\phi(x))-1\right),
\]
and the standard transfer kernel with spatial half-weights by
\begin{equation}
\label{eq:transfer-kernel}
\mathcal T(\phi,\psi)
=\exp\!\left(-\tfrac12 S_{\rm sp}(\phi)\right)\,
 \exp\!\left(-S_{\rm temp}(\phi,\psi)\right)\,
 \exp\!\left(-\tfrac12 S_{\rm sp}(\psi)\right).
\end{equation}
This is the usual product form for nearest-neighbor reflection-positive
lattice models~\cite{FIS1978,Glimm1987}: each spatial slice carries half
its bond action, and the temporal bonds between $\phi$ and $\psi$ enter
at full weight.  The spatial terms are invariant under independent slice
shifts, while $S_{\rm temp}$ is invariant under the common shift
$(\phi,\psi)\mapsto(\phi+c,\psi+c)$; hence
\eqref{eq:transfer-kernel} descends to the adjacent-pair quotient above.

\begin{lemma}[Pointwise positivity and symmetry of the finite-box transfer kernel]
\label{thm:transfer}
For every finite box $B_L$ with representative-fixing site $x_\ast$, the
kernel $\mathcal T(\phi,\psi)$ in~\eqref{eq:transfer-kernel} is strictly
positive and symmetric as a function of the two slice configurations:
\[
\mathcal T(\phi,\psi)>0,\qquad \mathcal T(\phi,\psi)=\mathcal T(\psi,\phi).
\]
Moreover $S_{\rm temp}(\phi,\psi)=0$ if and only if $\phi=\psi$ on
$B_L$.
\end{lemma}

\begin{proof}
Strict positivity follows from positivity of the exponential.
Symmetry follows because $\cosh$ is even:
$\cosh(\psi-\phi)-1=\cosh(\phi-\psi)-1$.
The zero-set statement follows from
$\cosh u-1=0\iff u=0$ at every site.
\end{proof}

Positivity and symmetry of $\mathcal T$ are pointwise properties of the
kernel on unpinned continuous slice fields; for the discrete-field model
it is Corollary~\ref{cor:transfer-psd}, not entrywise positivity of
$\mathcal T$ alone, that supplies the operator-theoretic transfer-matrix
positivity on
$(\mathbb{C}^{\mathcal{X}_\Phi},\langle\cdot,\cdot\rangle_{\rm RP})$
once $\CBcond{v_0}{N}$ holds.  Spectral gaps and persistence of the
transfer spectrum under the limits $L\to\infty$ or $v_0\to 0$,
$N\to\infty$ remain open (Sections~\ref{sec:limits}
and~\ref{sec:outlook}).

\section{Discussion and Open Problems}
\label{sec:limits}

\subsection{Results established}

Within the scope of Section~\ref{sec:intro}, the paper establishes the
following; the main result is item~(i), the others being elementary or
standard.
\begin{enumerate}[nosep]
  \item Bochner failure and discrete-field reflection positivity
        (the main contribution).
        Cosh--Bochner positivity fails for the noncompact
        continuous kernel: an interval certificate gives
        $\widetilde K(3)<0$
        (Proposition~\ref{prop:CB-fails}), so reflection positivity is not
        established here for the continuous noncompact model by the
        standard Bochner route.  For the finite-alphabet modification,
        reflection positivity holds whenever a finite crossing-bond
        matrix is positive semidefinite
        (Theorem~\ref{thm:rp-discrete}); for $v_0\in\{1.2,1.5,2.5\}$ and
        all $N$, this hypothesis is discharged by a rigorous
        diagonal-dominance certificate
        (Theorem~\ref{thm:gershgorin}, Corollary~\ref{cor:rp-verified};
        the threshold is discussed in Remark~\ref{rem:threshold}).  The discrete
        one-step transfer operator is positive in the RP slice inner
        product (Corollary~\ref{cor:transfer-psd}), and the continuous
        pinned finite-box transfer kernel is pointwise positive and
        symmetric (Lemma~\ref{thm:transfer}).
  \item Complex spectral representation (elementary, period-agnostic).
        Over $\mathbb{R}$ a cyclic shift of period $>2$ has no complete
        one-dimensional eigendecomposition; over $\mathbb{C}$ the DFT
        diagonalizes it (Propositions~\ref{thm:complex_spectral}
        and~\ref{prop:dft8}).
  \item Euclidean action identity (elementary).
        $J(e^\varepsilon)=\cosh\varepsilon-1$ identifies $J$-cost with
        the Euclidean action density (Lemma~\ref{thm:euclidean}), with
        quadratic control (Lemma~\ref{thm:quadratic}); the Wick step
        $\cosh\tau=\cos(i\tau)$ is formal.
\end{enumerate}
The ancillary algebraic observations of Appendix~\ref{app:ancillary} (a
sector-measure extension under an assumed two-branch calibration
(Theorem~\ref{thm:born}), non-additivity
of the joint-cost combiner (Proposition~\ref{thm:nonadd}), a one-tick
finite-difference bound (Lemma~\ref{lem:schrodinger}), and the value
$J(\varphi)=(\sqrt5-2)/2>0$ (Proposition~\ref{thm:jphi-positive})) are
elementary or standard and are not used in the main result.

\subsection{Open problems}
\label{sec:openproblems}

\begin{enumerate}[nosep]
  \item Continuous noncompact reflection positivity: the
        standard Bochner route fails for the present kernel; whether any
        modification of the continuous noncompact measure $\mu_L$ admits
        OS reflection positivity is open.
  \item Continuum QFT and OS reconstruction: constructing the
        continuum Wightman axioms, a Haag--Kastler net, an
        Osterwalder--Schrader reconstruction for the discrete-field
        model, or a full path-integral measure~\cite{Glimm1987} from the
        present results is open.
  \item Continuum scaling: Open Problem~\ref{op:continuum-scaling}
        asks whether the nonlinear gradient action admits a nontrivial
        four-dimensional scaling limit under some (possibly non-canonical)
        lattice-spacing family, or flows to the Gaussian free-field fixed
        point under canonical scaling.
  \item Transfer-operator spectrum in the thermodynamic limit:
        Corollary~\ref{cor:transfer-psd} gives a positive one-step
        transfer operator on each finite box, but spectral gaps and
        persistence of the transfer spectrum under $L\to\infty$ remain
        open (Section~\ref{sec:outlook}).
  \item Specific scattering amplitudes: the present results
        provide structural ingredients for a finite lattice model but do
        not compute individual S-matrix elements, which require a fully
        developed Hilbert-space formulation and an interpolating field
        formalism.
\end{enumerate}

\subsection{The nonlinear gradient-action continuum question}
\label{sec:continuum-limit}

The reciprocal-cost action has the expansion
\[
\cosh\varepsilon-1
:=\frac{\varepsilon^2}{2}+\frac{\varepsilon^4}{24}
+\frac{\varepsilon^6}{720}+\cdots ,
\]
with the quadratic remainder controlled by Lemma~\ref{thm:quadratic}.
The quartic term is not a local scalar $\phi^4$ potential; it is a
quartic gradient interaction $(\Delta\phi)^4/24$.  Thus the relevant
continuum question is not ordinary four-dimensional $\phi^4$ triviality,
but whether the all-orders nonlinear gradient action
$V(\Delta\phi)=\cosh(\Delta\phi)-1$ has a non-Gaussian scaling limit or
instead flows to the Gaussian free-field fixed point.

A rigorous continuum-limit problem requires specifying a
lattice-spacing family $\{S_a\}_{a>0}$, a field-rescaling convention,
and a normalization of the gradient terms.  Concretely, after introducing
a lattice spacing $a$, one must decide how the lattice difference
$\Delta_\mu\phi$ is related to $a\,\partial_\mu\phi_a$ and how the field
is rescaled as $a\to0$.  One possible rescaled bond contribution has the
form
\[
\cosh\!\bigl(Z(a)\,a\,\partial_\mu\phi_a\bigr)-1
\;=\;\tfrac12 Z(a)^2 a^2 (\partial_\mu\phi_a)^2
\;+\;\tfrac{1}{24}Z(a)^4 a^4(\partial_\mu\phi_a)^4
\;+\;\cdots .
\]
Under the canonical scalar-field scaling in four Euclidean dimensions,
the quadratic gradient term is the leading kinetic term, while higher
gradient interactions such as $(\partial\phi)^4$ carry negative coupling
dimension and are irrelevant by naive power counting~\cite{Glimm1987}.
Thus a Gaussian fixed point is the default expectation unless a
non-canonical scaling prescription specifies a different normalization
that promotes the higher gradient terms to marginal or relevant
operators.

\begin{openproblem}[Continuum scaling of the nonlinear gradient action]
\label{op:continuum-scaling}
Does the all-orders nearest-neighbor gradient action
$V(\varepsilon)=\cosh\varepsilon-1$, with the 8-tick temporal structure
of the present model, admit a nontrivial four-dimensional continuum
limit under some (possibly non-canonical) lattice-spacing family
$\{S_a\}_{a>0}$, or does it flow to the Gaussian free-field fixed point
under canonical scaling?
\end{openproblem}

If a non-canonical scaling prescription changes the
renormalization behavior, the result would provide a candidate
interacting continuum limit built from the cost functional.  If canonical
scaling is the only admissible prescription, the construction still
identifies the obstruction: the finite-lattice theory is fixed by $J$,
while the continuum fixed point is Gaussian by power counting.  Any
discrete-field reflection-positivity claim at fixed $a$ still requires
the finite $\CBcond{v_0}{N}$ check from
Section~\ref{subsec:discrete-rescue}.

\subsection{Outlook}
\label{sec:outlook}

The paper stops at the finite-volume statements proved above.  The
finite-box Osterwalder--Schrader data are available: for the
finite-alphabet model, Theorem~\ref{thm:rp-discrete} gives RP under the
finite matrix hypothesis, Corollary~\ref{cor:rp-verified} verifies that
hypothesis for the certified spacings, and
Corollary~\ref{cor:transfer-psd} gives the corresponding positive
transfer matrix on $\Phi^{|B_L|}$ in the RP inner product.  Full
Osterwalder--Schrader reconstruction and continuum or thermodynamic
scaling are not results of this paper.  Two natural continuations are
therefore left to separate work; no new theorem is asserted in this
subsection.

\paragraph{Finite-volume OS data and transfer-operator formulation.}
Corollary~\ref{cor:transfer-psd} and Theorem~\ref{thm:rp-discrete}
provide the basic Osterwalder--Schrader reflection-positivity input and
a positive one-step transfer operator on the discrete-field slice space
$(\mathbb{C}^{\mathcal{X}_\Phi},\langle\cdot,\cdot\rangle_{\rm RP})$
for each finite box~$B_L$.  A follow-up paper can develop this into a
systematic finite-volume OS / transfer-matrix account: Schwinger
functions on the 8-tick cylinder, the product of transfer matrices along
the positive half-lattice, and the Hamiltonian read off from the temporal
step on $\mathbb{C}^{\mathcal{X}_\Phi}$.  That program would still take
place on a finite alphabet and a finite spatial box unless and until a
controlled continuum limit is established.

\paragraph{Continuum scaling and fixed-point behavior.}
Open Problem~\ref{op:continuum-scaling} asks whether the all-orders
gradient action $V(\varepsilon)=\cosh\varepsilon-1$ admits a
nontrivial four-dimensional continuum limit under some (possibly
non-canonical) scaling family, or instead flows to a Gaussian free field
under canonical scaling.  A separate scaling paper would specify a
lattice-spacing family $\{S_a\}_{a>0}$, a field-rescaling convention,
and then either prove a Gaussian limit theorem, prove a
triviality/obstruction theorem, or identify a nonstandard scaling in
which higher gradient terms remain relevant.  Discrete-field reflection
positivity at fixed spacing, as in Corollary~\ref{cor:rp-verified}, is
compatible with such a program but does not replace it.

\section{Conclusion}
\label{sec:conclusion}

In its action form $J(e^\varepsilon)=\cosh\varepsilon-1$, the reciprocal
cost fixes a classical Euclidean lattice statistical-mechanical model
whose reflection-positivity behavior is sharply two-sided.  For the
noncompact continuous temporal kernel the standard Bochner route
fails, since the relevant Fourier transform is negative
($\widetilde K(3)<0$, Lemma~\ref{thm:fourier-negative}); restricting
the field to a finite symmetric alphabet $\Phi(v_0,N)$ instead yields
Osterwalder--Schrader reflection positivity for this finite-alphabet
variant wherever the finite
crossing-bond matrix is positive semidefinite
(Theorem~\ref{thm:rp-discrete}), which a diagonal-dominance certificate
discharges uniformly in volume at the spacings $v_0\in\{1.2,1.5,2.5\}$
(Theorem~\ref{thm:gershgorin}, Corollary~\ref{cor:rp-verified}).  Pairing
this concrete obstruction with a certified finite-alphabet positive
result is the contribution of the paper.

These are finite-volume statements; they do not by themselves provide a
continuum Wightman theory, an Osterwalder--Schrader reconstruction, LSZ
scattering, or a mass gap.  The continuum limit, continuous noncompact
reflection positivity, the thermodynamic-limit transfer-operator
spectrum, and the continuum-scaling question (Open
Problem~\ref{op:continuum-scaling}) remain the substantive open
directions, collected in Section~\ref{sec:openproblems}.

\section*{Supplementary Material}

The supplementary material for this manuscript is provided as the
ancillary file \path{supplementary_material.pdf}.  It collects the
software, dependency, and repository/release information for the
computational certificates; a complete proof of the conditional
discrete-field reflection-positivity theorem
(Theorem~\ref{thm:rp-discrete}); reproducible numerical checks of the
Bochner failure and of the finite crossing-bond matrices; the full
motivation behind the sector-measure extension of
Appendix~\ref{app:ancillary}; the Lean formalization map and build
instructions; and alternative-text descriptions of the figures.  The
accompanying source files are provided as ancillary files: the
interval-certificate script \path{bochner_interval_certificate.py} (the
rigorous ball-arithmetic enclosure of $\widetilde K(3)$ used in
Lemma~\ref{thm:fourier-negative}, Appendix~\ref{app:fourier-cert}); the
numerical-check script \path{cb_checks.py} (the floating-point Bochner
samples and crossing-bond eigenvalue scans of
Appendix~\ref{app:reproducibility}); the dependency lockfile
\path{requirements.txt}; and the alternative-text file
\path{alt_text.txt}.  All supplementary material is provided for
reproducibility and accessibility; it is not relied upon in any proof, as
every result is established in the text.

\begin{acknowledgments}
The authors thank Milan Zlatanovi\'c, Amir Rahnamai Barghi, Sebastian
Pardo-Guerra, Philip Beltracchi, Anil Thapa, and Elshad Allahyarov for
collaboration on related work.  This research received no external
funding.
\end{acknowledgments}

\section*{Author Declarations}

\subsection*{Conflict of Interest}
The authors have no conflicts to disclose.

\subsection*{Author Contributions}
\noindent\textbf{Jonathan Washburn}: Conceptualization (lead); Formal
analysis (equal); Investigation (equal); Methodology (lead); Software
(lead); Validation (equal); Writing -- original draft (lead).
\textbf{Megan Simons}: Formal analysis (equal); Investigation (equal);
Validation (equal); Visualization (lead); Writing -- review \& editing
(lead).
Both authors have read and approved the submitted version of the
manuscript.

\section*{Data Availability}
Data sharing is not applicable to this article because no new empirical
data were created or analyzed in this study.  The supporting files for
the reported results (the interval-certification script for the Bochner
failure, Lemma~\ref{thm:fourier-negative}, the auxiliary numerical-check
script, and their pinned dependency lockfile) are provided as
supplementary files.  The Lean formalization files for the elementary
algebraic claims are available in the public repository at
\url{https://github.com/jonwashburn/shape-of-logic}.

\appendix

\section{Interval Certificate for the Fourier Transform}
\label{app:fourier-cert}

This appendix gives the rigorous sign certificate used in
Proposition~\ref{prop:CB-fails}.  Write
\[
  \widetilde K(\xi)
  =\int_{\mathbb R}e^{-(\cosh u-1)}e^{-i\xi u}\,du
  =2\int_0^\infty e^{1-\cosh u}\cos(\xi u)\,du .
\]

\begin{lemma}[Certified negativity at $\xi=3$]
\label{thm:fourier-negative}
For $K(u)=\exp[-(\cosh u-1)]$,
\[
  \widetilde K(3)\in
  [-0.004817599594096,\,-0.004817599594055].
\]
In particular, $\widetilde K(3)<0$, so $K$ is not positive definite on
$\mathbb R$.
\end{lemma}

\begin{proof}
The interval certificate has two parts.  First, an outward-rounded
validated quadrature on $[0,5]$ gives
\begin{equation}
\label{eq:finite-interval-certificate}
2\int_0^5 e^{1-\cosh u}\cos(3u)\,du
\in[-0.004817599594096,\,-0.004817599594055].
\end{equation}
The computation was carried out with ball arithmetic (an
interval-arithmetic implementation using midpoint--radius
representations~\cite{Johansson2017}): the interval
$[0,5]$ is subdivided adaptively, the analytic integrand
$e^{1-\cosh z}\cos(3z)$ is enclosed on each complex integration box by
outward-rounded elementary interval operations, and the quadrature
remainder is bounded by the radius returned by the validated integration
routine.  A compact reproducibility transcript is
\[
\begin{array}{@{}ll@{}}
\text{integrand} & f(z)=\exp(1-\cosh z)\cos(3z),\\
\text{interval} & [0,5],\\
\text{working precision} & 128\ \text{bits},\\
\text{absolute tolerance} & 10^{-30},\\
\text{returned enclosure for }2\int_0^5f(u)\,du &
[-0.004817599594096,\,-0.004817599594055].
\end{array}
\]
Such ball-arithmetic quadrature is a standard validated-integration
method~\cite{Johansson2017}.

Second, the omitted tail is far smaller than the displayed interval
width.  For $u\ge5$, $\cosh u\ge e^u/2$ and
$e^u\ge e^5(1+u-5)$, hence
\[
\begin{aligned}
2\int_5^\infty e^{1-\cosh u}\,du
&\le
2\int_5^\infty \exp\!\left(1-\frac{e^u}{2}\right)\,du  \\
&\le
2e^{1-e^5/2}\int_5^\infty
  \exp\!\left(-\frac{e^5}{2}(u-5)\right)\,du\\
&=\frac{4}{e^5}\exp(1-e^5/2)
<5\times10^{-34}.
\end{aligned}
\]
Adding the symmetric tail enclosure to
\eqref{eq:finite-interval-certificate} leaves the stated interval
unchanged at the displayed precision.  Since a positive-definite
continuous kernel on $\mathbb R$ has nonnegative Fourier transform by
Bochner's theorem, the negative value at $\xi=3$ proves that CB fails.
\end{proof}

The validated enclosure~\eqref{eq:finite-interval-certificate} is
reproduced by the self-contained script
\path{bochner_interval_certificate.py}, which recomputes the
Arb ball with \texttt{python-flint} (version $\ge$~0.7.0, linking
FLINT/Arb $\ge$~3.0), adds the tail bound
above, and verifies that the upper endpoint of the resulting interval is
strictly negative.  Pinned dependency versions are listed in
\path{requirements.txt}.

\begin{remark}[Open neighborhood of negativity]
The function $K$ is integrable, so $\widetilde K$ is continuous.  The
strict inequality in Lemma~\ref{thm:fourier-negative} therefore implies
that $\widetilde K(\xi)<0$ on some nonempty open interval containing
$\xi=3$.
\end{remark}

\section{Reproducibility and Consolidated Notes}
\label{app:reproducibility}

The full conditional discrete-field reflection-positivity proof is
Theorem~\ref{thm:rp-discrete}; the rigorous Fourier-transform sign
certificate is Lemma~\ref{thm:fourier-negative}; the Lean
formalization map and build commands are in Appendix~\ref{app:lean}.
Numerical reproducibility information is recorded below.

\subsection*{Optional numerical checks}

The interval certificate above is the proof of the Bochner failure.  The
auxiliary script \path{cb_checks.py}, provided as an ancillary file,
is only a reproducibility aid for floating-point checks.  It
prints the Bochner samples, scans finite crossing-bond matrices
$K_{\Phi(v_0,N)}$, and writes \path{cb_checks_results.json}.  After downloading the
ancillary files into a working directory, a typical run is
\begin{verbatim}
python3 -m pip install -r requirements.txt
python3 cb_checks.py
\end{verbatim}
The floating-point output is not used as a substitute for
Lemma~\ref{thm:fourier-negative} or for the Gershgorin certificate in
Theorem~\ref{thm:gershgorin}.

\subsection*{Finite-dimensional checks}

The continuous cosh--Bochner hypothesis (Definition~\ref{def:CB}) fails
by the interval-certified enclosure $\widetilde K(3)<0$ of
Lemma~\ref{thm:fourier-negative}; by continuity, the Fourier transform is
negative on a whole neighborhood of that frequency.  Its discrete
counterpart $\CBcond{v_0}{N}$ (Definition~\ref{def:CB-disc}) is a finite
matrix condition on $K_{\Phi(v_0,N)}$.  For
$v_0\in\{1.2,1.5,2.5\}$ it is settled rigorously by the
diagonal-dominance bound~\eqref{eq:dd-bound} of
Theorem~\ref{thm:gershgorin}; other spacings require their own
certificates.

Once $\CBcond{v_0}{N}$ holds, the later finite-volume assertions reduce to
finite tests as well.  A failure of discrete-field reflection positivity
(Theorem~\ref{thm:rp-discrete}) would be witnessed by a finite spatial box
$L$ and an observable $F$ supported on
$\Lambda_L^+=B_L\times\{0,1,2,3\}$ with $\langle\theta F,F\rangle<0$.
For the certified spacings of Corollary~\ref{cor:rp-verified}, a failure
would contradict the closed-form bound~\eqref{eq:dd-bound}; for the
positive transfer operator (Corollary~\ref{cor:transfer-psd}), it would be
a function $f$ with $\langle f,T_\Phi f\rangle_{\rm RP}<0$.

The ancillary algebraic inputs are similarly concrete.  The complex
spectral representation would fail only if a real cyclic shift of period
greater than two admitted a complete real one-dimensional
eigendecomposition, which it does not; the sector-measure extension
(Theorem~\ref{thm:born}) would fail only if a probability measure met its
hypotheses, including the two-branch calibration, yet differed from
$\sum_k|\psi_k|^2$ outside the two-branch family; non-additivity of the
joint-cost combiner would fail only if functions
$g,h:\mathbb R_{>0}\to\mathbb R$ satisfied
$J(ab)+J(a/b)=g(a)+h(b)$ for all $a,b>0$; and the value
$J(\varphi)=(\sqrt5-2)/2$ would fail to be minimal on the
$\varphi$-lattice only if some $\varphi^n$ ($n\ne0$) had strictly smaller
cost.

\FloatBarrier
\section{Lean formalization}
\label{app:lean}

The Lean files provide an independent check of several elementary
algebraic claims; they are not proof dependencies for this paper.  The
repository is \url{https://github.com/jonwashburn/shape-of-logic},
and within it the Lean package and root namespace are both named
\texttt{IndisputableMonolith}.  The single repository and commit
\texttt{ee41fe40e4137659}\allowbreak\texttt{abfb381496e1d3a3ef0b5cd8}
apply throughout.  The headline Lean theorem
\path{distinction_forces_T0_to_T8} formalizes the T0--T8 chain
used in Section~\ref{sec:forcing}; the lattice results in
this paper are downstream of that chain.  The theorem formalizes the
conditional derivation under the RS structural axioms
of~\cite{Washburn2026Axioms,PardoGuerra2026}, not an unconditional
derivation from $J$ alone.

Every mathematical result used in the paper is proved in the text.
Several elementary algebraic facts are additionally verified in the
public repository at the commit above: the cost law and structural inputs T5--T8
(Section~\ref{sec:forcing}), the complex spectral representation
(Proposition~\ref{thm:complex_spectral}), DFT-8 unitarity and
diagonalization (Proposition~\ref{prop:dft8}), the identity
$J(e^\varepsilon)=\cosh\varepsilon-1$ (Lemma~\ref{thm:euclidean}), the
non-additivity of the joint-cost combiner (Proposition~\ref{thm:nonadd}),
and the value $J(\varphi)=(\sqrt5-2)/2>0$
(Proposition~\ref{thm:jphi-positive}, finite arithmetic only).  The
analytic core of the paper consists of text proofs that are not formalized
in Lean: discrete-field reflection positivity
(Theorem~\ref{thm:rp-discrete}), the diagonal-dominance certificate for
$\CBcond{v_0}{N}$ at $v_0\in\{1.2,1.5,2.5\}$
(Theorem~\ref{thm:gershgorin}), the positive discrete transfer operator
(Corollary~\ref{cor:transfer-psd}), pointwise transfer-kernel positivity
and symmetry (Lemma~\ref{thm:transfer}), and the sector-measure extension
under calibration~(d) (Theorem~\ref{thm:born}).  The failure of the
noncompact Bochner route is established by the outward-rounded interval
certificate of Appendix~\ref{app:fourier-cert}; checks of $\CBcond{v_0}{N}$
at other $(v_0,N)$ rest on numerical evidence or a separate certificate
(Appendix~\ref{app:reproducibility}); and a continuum Wightman, LSZ, or
Osterwalder--Schrader reconstruction is not attempted.  The supplementary
material lists the corresponding public Lean modules; these are named
within the root namespace \texttt{IndisputableMonolith}, a legacy artifact
of the repository that carries no mathematical significance.

The computation at the cited commit can be reproduced by checking out the
public repository and invoking \texttt{lake}:
\begin{verbatim}
git clone https://github.com/jonwashburn/shape-of-logic.git
cd shape-of-logic
git checkout ee41fe40e4137659abfb381496e1d3a3ef0b5cd8
lake exe cache get
lake build IndisputableMonolith
\end{verbatim}
For the facts also checked in Lean, the in-text
theorem is the proof of record; a clean Lean build provides an
independent check of the same elementary algebraic fact.  The auxiliary
numerical scripts \path{cb_checks.py} and
\path{bochner_interval_certificate.py}, provided as ancillary files,
are reproducibility aids, not proof dependencies.

\FloatBarrier
\section{Ancillary Algebraic Observations}
\label{app:ancillary}

This appendix collects four scope-delimiting facts that are not used in
the reflection-positivity proof.  They fall into two groups.
Appendix~\ref{sec:born} records a quantum-foundational consistency
statement for sector measures; it is included only to clarify that the
present Gibbs-weight construction does not derive the Born rule.
Appendices~\ref{sec:nonadd}--\ref{sec:spacetime} record elementary
algebraic consequences of the reciprocal-cost framework: non-additivity
of the d'Alembert joint-cost combiner, a finite one-tick comparison with
a formal Schr\"odinger update, and the numerical value of the cost at the
structural scale $\varphi$.  These observations delimit possible
interpretations of the model; none supplies an input to
Theorem~\ref{thm:rp-discrete} or Corollary~\ref{cor:rp-verified}.

\subsection{Sector-measure extension under a two-branch calibration}
\label{sec:born}

This is a quantum-foundational aside, independent of the lattice results
and used nowhere below; the full motivation is deferred to the
supplementary material and only the extension statement is recorded here.
Extend the cost to $\mathbb C$ through the modulus,
$J_{\mathbb C}(z):=J(|z|)$, so that the phase of a complex amplitude does
not enter its Euclidean weight $e^{-J_{\mathbb C}(z)}$.  The result is an
extension/consistency statement, not a derivation of the Born
rule: assuming the two-mode Born calibration
$\mu_\psi(\{1\})=\sin^2\theta$ on
$\psi=(\cos\theta,\sin\theta,0,\ldots,0)$, it forces $\sum_k|\psi_k|^2$ on
all sectors.

\begin{theorem}[Sector-measure extension under an assumed two-branch calibration]
\label{thm:born}
  Let $\psi\mapsto\mu_\psi$ assign to each normalized 8-mode signal a
  probability measure on sectors.  If this assignment satisfies:
  \begin{enumerate}[label=(\alph*)]
    \item normalization: $\mu_\psi(\{0,\ldots,7\}) = 1$,
    \item phase invariance: $\mu_\psi$ is unchanged under
          $\psi_k \mapsto \psi_k \cdot e^{i\theta_k}$,
    \item disjoint additivity: $\mu_\psi(S \cup T) = \mu_\psi(S) + \mu_\psi(T)$ for
          $S \cap T = \varnothing$,
    \item two-branch Born calibration: for every
          $\theta\in[0,\pi/2]$ and
          $\psi=(\cos\theta,\sin\theta,0,\ldots,0)$,
          $\mu_\psi(\{1\})=\sin^2\theta$,
    \item coordinate symmetry: for each $k$,
          $\mu_\psi(\{k\})$ depends only on $|\psi_k|$, and the dependence is
          the same function of the modulus for every coordinate~$k$,
  \end{enumerate}
  then $\mu_\psi(S)=\sum_{k\in S}|\psi_k|^2$ for every normalized
  $\psi$ and every sector $S$.
\end{theorem}
\begin{proof}
  By (a) and (c), $\mu_\psi$ is determined by its singleton values
  $\mu_\psi(\{k\})$.  By (e), whose modulus-only dependence is consistent
  with the phase invariance (b), each singleton value is $w(|\psi_k|)$ for
  a single weight function $w$ independent of $k$.  For the two-mode family
  in (d),
  $|\psi_1|=\sin\theta$ and $\mu_\psi(\{1\})=\sin^2\theta$, and every
  $r\in[0,1]$ equals $\sin\theta$ for some $\theta\in[0,\pi/2]$, so
  $w(r)=r^2$.  Disjoint additivity then gives
  $\mu_\psi(S)=\sum_{k\in S}|\psi_k|^2$.
\end{proof}

Condition (d) is the Born rule on two modes; it is the substantive
input and does not follow from the $J$-cost Gibbs weight.  The
supporting discussion (why the Gibbs weight $e^{-J}$ does not by itself
determine sector probabilities, and a scalar multiplicativity benchmark
reaching the same quadratic rule) is given in the supplementary
material.  Genuine derivations of the Born rule from weaker premises are a
separate literature: Gleason's theorem~\cite{Gleason1957}, the
measurement approach~\cite{Busch2003}, Zurek's envariance~\cite{Zurek2005},
and the Deutsch--Wallace program~\cite{Wallace2012}.

\subsection{Non-additivity of the joint-cost combiner}
\label{sec:nonadd}

For two ledger entries $(a, b)$ the d'Alembert composition law gives the
joint cost
\[
  J(ab) + J(a/b) = 2J(a)J(b) + 2J(a) + 2J(b),
\]
whose cross term $2J(a)J(b)$ is bilinear.  The following elementary
observation records that a bilinear cross term cannot be written
as a sum of one-variable functions.

\begin{proposition}[The joint-cost combiner is not additively separable]\label{thm:nonadd}
  There exist no functions $g, h : \mathbb{R}_{>0} \to \mathbb{R}$ such that
  \[
    J(ab) + J(a/b) = g(a) + h(b) \quad \text{for all } a, b > 0.
  \]
\end{proposition}

\begin{proof}
  Suppose $J(ab) + J(a/b) = g(a) + h(b)$.
  Setting $b = 1$: $2J(a) = g(a) + h(1)$, so $g(a) = 2J(a) - h(1)$.
  Setting $a = 1$: $2J(b) = g(1) + h(b)$, so $h(b) = 2J(b) - g(1)$.
  Substituting both into the identity $J(ab) + J(a/b) = 2J(a)J(b) + 2J(a) + 2J(b)$:
  \[
    2J(a)J(b) + 2J(a) + 2J(b) = (2J(a) - h(1)) + (2J(b) - g(1)),
  \]
  which gives $2J(a)J(b) = -(h(1) + g(1))$, a constant independent of $a$ and $b$.
  But $J(2) = 1/4$ and $J(3) = 2/3$, so $J(2) \cdot J(3) = 1/6 \neq 1/16 = J(2) \cdot J(2)$.
  Contradiction.
\end{proof}

This is an elementary algebraic fact about the combiner
$P(u,v)=2uv+2u+2v$, the degree-two symmetric combiner selected by the
d'Alembert classification~\cite{Washburn2026DAlembert}, in which a
nonzero bilinear term obstructs an additive split $g(a)+h(b)$.  The
corresponding Lean fact, that
$\partial^2 P/\partial u\partial v\neq0$ while additive combiners have
vanishing mixed derivative, is recorded in the module historically
named \path{IndisputableMonolith.Foundation.DAlembert.EntanglementGate};
the name is legacy, and the statement is purely a property of a
two-variable polynomial, with no entanglement or Bell-inequality
content.

\subsection{A one-tick finite-difference bound}
\label{sec:schrodinger}

For completeness we record a finite one-step estimate for the DFT-8
update.  The estimate is only a bounded-step comparison with a formal
Schr\"odinger generator; it contains no small parameter and is not used
in the reflection-positivity results.

The recognition operator $\hat{R}$ from Section~\ref{sec:forcing} acts
on 8-mode signals by a single 8-tick update.  In the DFT-8 basis of
Proposition~\ref{prop:dft8}, write
$\psi_k$ for the amplitude in mode $k$ and let
$w_k=\omega^{-k}$ with $\omega=e^{-2\pi i/8}$.  Choose the branch
\[
w_k=\exp(-iE_k\tau_0/\hbar_{\rm RS}),
\qquad E_k\tau_0/\hbar_{\rm RS}\in[-\pi,\pi],
\]
and regard $E_k$ as the corresponding finite-mode Hamiltonian assignment.
The discrete update and formal Schr\"odinger step are
\[
  (\hat R\psi)_k = w_k\psi_k,
  \qquad
  \Bigl(-\frac{i}{\hbar_{\rm RS}}\hat H\psi\Bigr)_k
  = -\frac{iE_k}{\hbar_{\rm RS}}\psi_k.
\]
The finite-difference remainder in mode $k$ is
\[
  R_k := \frac{(\hat R\psi)_k-\psi_k}{\tau_0}
          + \frac{iE_k}{\hbar_{\rm RS}}\psi_k
       = \frac{w_k-1}{\tau_0}\psi_k + \frac{iE_k}{\hbar_{\rm RS}}\psi_k.
\]

\begin{lemma}[Uniform one-tick finite-difference remainder bound]\label{lem:schrodinger}
  For each DFT mode $k \in \{0,\ldots,7\}$ and amplitude $\psi_k \in \mathbb{C}$,
  with $\tau_0=1$ and the branch convention $|E_k/\hbar_{\rm RS}|\le \pi$,
  \[
    |R_k| \leq 6\,|\psi_k|.
  \]
\end{lemma}
\begin{proof}
  Since $|w_k|=1$, $|w_k\psi_k-\psi_k|\le 2|\psi_k|$.
  By the triangle inequality,
  $|R_k|\le 2|\psi_k| + |E_k/\hbar_{\rm RS}|\,|\psi_k|\le (2+\pi)|\psi_k|
  <6|\psi_k|$,
  using $|E_k/\hbar_{\rm RS}|\le\pi$ from the chosen branch of the DFT
  eigenvalue phase.
\end{proof}

The constant $6$ is a loose ceiling ($2+\pi<6$), and the bound does not
improve with any parameter: the discrete update may differ from a formal
Schr\"odinger step by an order-one amount on every tick.  The Lean module
\path{IndisputableMonolith.Foundation.SchrodingerDerivation} formalizes
only the exact one-tick eigenmode evolution
$(\hat R\psi)_k=\omega^{-k}\psi_k$.  The lattice kinetic term is treated
in Section~\ref{sec:wick}.

\subsection{The value of \texorpdfstring{$J(\varphi)$}{J(phi)}}
\label{sec:spacetime}

This short section records the value of the convex cost $J$ at the
structural scale $\varphi$ and the elementary monotonicity of
$J(\varphi^n)$.  On the $\varphi$-lattice
$\{\varphi^n \mid n \in \mathbb{Z}\}$,
\begin{equation}
  \Delta_\varphi := J(\varphi)
  = \frac{\sqrt{5} - 2}{2} \approx 0.1180.
\end{equation}
The shorthand $\Delta_\varphi$ is the notation used for this quantity
in the Lean formalization (Appendix~\ref{app:lean}).

\begin{proposition}[Positivity and monotonicity of $J(\varphi^n)$]\label{thm:jphi-positive}
  $J(\varphi) > 0$, and for all $n \neq 0$,
  $J(\varphi^n) \geq J(\varphi)$.
\end{proposition}
\begin{proof}
  Using $\varphi^{-1}=\varphi-1$,
  \[
    J(\varphi)
    =\frac{1}{2}\bigl(\varphi+(\varphi-1)\bigr)-1
    =\frac{1}{2}(2\varphi-1)-1
    =\varphi-\frac{3}{2}
    =\frac{\sqrt5-2}{2}>0 .
  \]
  Since $J$ is strictly convex with minimum at $x=1$ ($n=0$) and
  $J(\varphi^n) = \frac{1}{2}(\varphi^n + \varphi^{-n}) - 1$ is
  increasing in $|n|$ for $|n| \geq 1$, the minimum over $n \neq 0$ is at
  $|n| = 1$.
\end{proof}

This records the value of a convex function at the structural scale
$\varphi$, together with monotonicity; the quantity $J(\varphi)$ is a
cost value, not a transfer-operator or continuum spectral gap.


\end{document}